\documentclass[11pt]{amsart}

\usepackage{amsthm,amssymb, amsmath}
\usepackage[author-year]{amsrefs}
\usepackage{color}
\usepackage[normalem]{ulem}
\usepackage{graphicx}
\usepackage{enumerate}
\usepackage[utf8]{inputenc}
\usepackage{amsaddr}

\definecolor{labelcolor}{RGB}{100,0,0}
\definecolor{blue-violet}{rgb}{0.54, 0.17, 0.89}
\newcommand{\changed}[1]{{#1}}
\newcommand{\changedA}[1]{{#1}}
\newcommand{\changedB}[1]{}

\newcommand\redout{\bgroup\markoverwith
{\color{red}{\rule[.5ex]{2pt}{0.4pt}}}\ULon}

\DeclareMathOperator{\sgn}{sgn}

\newcommand{\opload}[1]{\mathrm{const}_{#1}}
\newcommand{\opmult}[2]{\mathrm{mult}_{#1,#2}}
\newcommand{\opdiv}[2]{\mathrm{div}_{#1,#2}}
\newcommand{\opadd}[2]{\mathrm{add}_{#1,#2}}
\newcommand{\opsub}[2]{\mathrm{sub}_{#1,#2}}
\newcommand{\opcpy}[1]{\mathrm{cpy}_{#1}}
\newcommand{\shiftl}{\text{shift-left}}
\newcommand{\shiftr}{\text{shift-right}}

\newcommand{\Prec}{\mathrm{prec}}
\newcommand{\Arith}{\mathrm{arith}}

\newcommand{\Length}[1]{\mathrm{Length}{(#1)}}
\newcommand{\Size}[1]{\mathrm{Size}{(#1)}}

\newcommand{\SAFeas}{\text{\small \sc SA-Feas}}
\newcommand{\SAEFeas}{\text{\small \sc SAE-Feas}}
\newcommand{\CircFeas}{\text{\small \sc Circ-Feas}}
\newcommand{\CircPseudoFeas}{\text{\small \sc Circ-Pseudo-Feas}}

\newcommand{\fl}{\mathrm{fl}}

\renewcommand{\P}{\ensuremath{\mathbf{P}} }
\newcommand{\realP}{\ensuremath{\mathbf{P}_{\mathbb R}} }
\newcommand{\realPone}{\ensuremath{\mathbf{P}'_{\mathbb R}} }
\newcommand{\realPtwo}{\ensuremath{\mathbf{P}''_{\mathbb R}} }
\newcommand{\fpP}{\ensuremath{\mathbf{P}_{\mathbb R}^{\mathrm{fp}}} }

\newcommand{\NP}{\ensuremath{\mathbf{NP}} }
\newcommand{\realNP}{\ensuremath{\mathbf{NP}_{\mathbb R}} }
\newcommand{\realNPone}{\ensuremath{\mathbf{NP}'_{\mathbb R}} }
\newcommand{\fpNP}{\ensuremath{\mathbf{NP}_{\mathbb R}^{\mathrm{fp}} }}

\newcommand{\realEXPone}{\ensuremath{\mathbf{EXP}'_{\mathbb R}}}
\newcommand{\fpEXP}{\ensuremath{\mathbf{EXP}^{\mathrm{fp}}}}

\newcommand{\PR}{\ensuremath{\mathbf{P}_{\mathbb R}} }
\newcommand{\co}{\ensuremath{\mathbf{co-}}}
\newcommand{\Pro}{\ensuremath{\mathbf{P_{ro}}} }
\newcommand{\NPro}{\ensuremath{\mathbf{NP_{ro}}} }
\newcommand{\Piter}{\ensuremath{\mathbf{P_{iter}}} }

\newcommand{\NPiterU}{\ensuremath{\mathbf{NP_{iter}^U}} }
\newcommand{\NPiterB}{\ensuremath{\mathbf{NP_{iter}^B}} }
\newcommand{\umach}{\ensuremath{u_{\mathrm{mach}}}}

\newcommand{\dd}{\ensuremath{\ {\mathrm d}}}
\newcommand{\function}[5]{\ensuremath{
\begin{array}{rccl}
#1 : & #2 & \rightarrow & #3,\\ & #4 & \mapsto & #5
\end{array}}}

\newtheorem{main}{Theorem}

\newtheorem{theorem}{Theorem}
\numberwithin{theorem}{section}
\newtheorem{lemma}[theorem]{Lemma}
\newtheorem*{lemma*}{Lemma}
\newtheorem{proposition}[theorem]{Proposition}

\theoremstyle{definition}

\newtheorem{definition}[theorem]{Definition}

\newtheorem{example}[theorem]{Example}

\theoremstyle{remark}
\newtheorem{remark}[theorem]{Remark}

\author{Gregorio Malajovich}
\address{Instituto de Matemática,\\Universidade Federal do Rio de Janeiro,\\Caixa Postal 68530, Rio de Janeiro RJ 21941-909 Brazil\\gregorio.malajovich@gmail.com}
\author{Mike Shub}
\address{
Department of Mathematics,\\
The City College of the City University of New York.\\
NAC 8/133\\
160 Convent Avenue, New York, NY 10031\\
mshub@ccny.cuny.edu
}
\date{\changedA{Revised version. Jan 16, 2019}}
\title[A theory of NP-completeness and ill-conditioning...]{A theory of NP-completeness and ill-conditioning for approximate real computations}
\thanks{\changedA{2012 {\em ACM Computing Classification System.}
Complexity classes; Problems, reductions and completeness.}}
\thanks{This work was partially supported by the Smale Institute.}

\begin{document}
\maketitle
\begin{abstract}
We develop a complexity theory for approximate real computations. 
We first produce a theory for exact computations but with condition numbers.
The input size depends on a condition number, which is not assumed known by the machine. The theory admits deterministic and nondeterministic polynomial time recognizable problems. We prove that P is not NP in this theory if and only if P is not NP in the BSS theory over the reals. 
\par
Then we develop a theory with weak and strong approximate computations. This theory is intended to model actual numerical computations that are usually performed in floating point arithmetic. It admits classes P and NP and also an NP-complete problem.
We relate the P vs NP question in this new theory to the classical P vs NP problem.  
\end{abstract}

\tableofcontents
\section{Introduction}
\ocite{BSS} proposed a model of computation over the real and
complex numbers and \changedA{over} any ring or field. The initial goal was to provide
foundations and a theory of complexity for numerical analysis and scientific
computing. They borrowed components from several existing theories, especially algebraic complexity and the complexity theory of theoretical computer science including the \P versus \NP problem. The model was essentially a Turing machine in which the entries on the tape would be elements of the real numbers, complex numbers or the other rings or fields in question and the arithmetic operations and comparisons exact. It was intended to capture essential components of computation as employed in numerical analysis and scientific computing in a simple model which would put the theory in contact with additional mainstream areas of mathematics and produce meaningful or even predictive results.
The need to incorporate approximations, input and round-off error was mentioned at the time but there has not been much progress in those directions in the intervening years especially as concerns decision problems. In this paper we propose a remedy for this situation.
\par
A good part of the problem concerns how to describe the input error and \changedA{the} round-off error. Here we choose floating point arithmetics. 
Our model has the following properties which we refer to as the {\em wish list}:
\begin{enumerate}[(a)]
	\item The theory admits \changedA{classes $\P$ and $\NP$} with $\P \subseteq \NP$.
	\item The class \P contains the class \P of classical (Turing) computations of computer science. Moreover it contains decision sets that are related to computations considered easy in numerical practice, such as the complements of graphs of elementary functions. It also contains problems related to standard linear algebra computations and certain fractal sets.
	\item The class \NP contains a complete problem.
\item Machines supporting the definitions of \P and \NP never give wrong
	answers, regardless of the precision. 
\item Numerical stability issues do play a role.
\item The condition number plays a major role in the theory.
\end{enumerate}
We do not require \P to be closed by complements, as this seems to preclude other important goals. The condition number used is quite general as the one used by \ocite{Cucker}. This definition emerged from discussions between Cucker and the first author of this paper. This generality is useful as a natural definition of condition as the reciprocal distance to the locus of ill-posedness varies according to the context. 
\par
\subsection*{Outline of the paper and main results}
We proceed in two steps. First we define classes \realPone and \realNPone of real decision
problems with a condition number. Those classes generalize the classes $\realP$ and $\realNP$ as in \ocite{BSS} to include condition numbers, but still with exact computations. This will allow us to investigate all the main
features of the theory, except for numerical stability. The main result in Part~1 will be
\begin{main}
The following are equivalent:
	\begin{enumerate}[(a)]
		\item $\realP \neq \realNP$.
		\item $\realPone \neq \realNPone$.
	\end{enumerate}
\end{main}

This theorem is restated below as Theorem \ref{toy-transfer}. In Part 2 we will introduce
a model of floating point computations, and classes \fpP and \fpNP . 
The main results in Part 2 are the existence of an $\fpNP$-complete problem and the Theorem below, restated as Theorem \ref{real-transfer}.
\begin{main}
If $\P \neq \NP$, then $\fpP \neq \fpNP$.
\end{main}

\subsection*{Related work}
Alan Turing \ycite{Turing} understood that the main obstruction to the
efficiency of numerical computations would be loss of accuracy due to iterated
round-off errors. He realized that
\begin{quotation}
	\em
When we come to make estimates of errors in matrix processes we
shall find that the chief factor limiting the accuracy that can be
obtained is ‘ill-conditioning’ of the matrices involved.
\end{quotation}
This motivated the introduction of the condition number as a ‘measure of 
ill-conditionning’.

Formal models of computability for real functions were developed later. Recursive Analysis is an extension of Turing’s model of computation to Turing machines
with infinite tapes for input and output. Those machines can be used to compute
maps between topological spaces with a countable basis. The input is a convergent
nested sequence of balls from the basis of input space, and the output is the same
for output space. This model has the property that only continuous functions can
be computable. In particular, decidable sets must be both open and closed.

An attempt to propose a more realistic model of numerical computation was
made by \ocite{BSS}. This model admits a real $\NP$-completeness theory
similar to the classical theory by Cook and Karp. Condition numbers and rounding-off are not incorporated into the BSS model. A strong objection against it was raised by \ocite{Braverman-Cook}:
\begin{quotation}
	\em
	A weakness of the BSS approach as a model of scientific computing is that uncomputability results do not correspond to computing practice in the case $R=\mathbb R$.
\end{quotation}
This objection was further elaborated by \ocite{Braverman-Yampolsky}*{p.13}:
\begin{quotation}
	\em
	Algebraic in nature, BSS decidability is not well-suited for the study
of fractal objects, such as Julia sets. It turns out (see Chapter 2.3
of (Blum, Cucker, Shub, and Smale, 1998)) that sets with a 
fractional Hausdorff dimnension, including ones with very simple 
description, such as Cantor set and the Koch snowflake (...), are
BSS-undecidable. Morevoer, due to the algebraic nature of the
model, very simple sets that do not decompose into a countable
union of semi-algebraic sets are not decidable. An example of such
	a set is the graph of the function $f(x)=e^x$ (...)
\end{quotation}

They proposed instead a theory of sets recognizable in polynomial time, but without
an $\NP$-completeness theory. The class \realPone that we propose borrows 
from the idea
of measuring the cost of recognizing a set (or its complement), instead of 
the cost of deciding it.

\changed{A first tentative to endow the BSS model with condition numbers
and approximate computations \changedA{is} due to \ocite{Cucker-Smale}\label{ref01}. They
studied specifically algorithms for deciding semi-algebraic sets under
an absolute error model for numerical computations, but no reduction 
theory was developed. }

\ocite{Cucker}*{Remark 7} defined a model of numerical computation with condition numbers similar to the one in this paper. In his model a problem is a pair
$(X, \mu)$ where \changed{$X \subset \mathbb R^{\infty}$ and
the notation $\mathbb R^{\infty}$\label{ref02} stands for the disjoint union of all $\mathbb R^k$,
$k \in \mathbb N$. The {\em condition number} 
$\mu: \mathbb R^{\infty} \rightarrow [1, \infty]$}
is an arbitrary function. We will retain
this definition. The size of an input $x$ is its length plus $\log(\mu(x))$. We will use a
similar definition (length times $\log(\mu(x))$) that preserves the polynomial hierarchy.
A machine in Cucker’s model is a BSS machine over $\mathbb R$ modified so that all
computations are approximate, and for instance $c \leftarrow a+b$  produces a real number
$c$ so that
\[
	(1 - \epsilon)(a + b) \le c \le (1 + \epsilon)(a + b).
\]
The number $\epsilon$ is known to the machine. The cost of a computations is 
$-\log(\epsilon)$ times
the number of arithmetic operations and branches.
Several classes are defined. The classes \Pro and \NPro allow for a machine to
give a wrong answer if the precision $\epsilon$ is not small enough. This fails one of our
main wishes in this paper. An example of undesirable consequences is a constant
time algorithm to decide the problem $(K,\mu)$ where $K$ is Cantor’s middle thirds set
and $\mu(x) = d(x, K)-1$ is the natural condition number (see example \ref{Cantor}).
Cucker also defined a class where machines are not supposed to give wrong
answers. Those are classes \Piter and \NPiterU. Quoting from 
\ocite{Cucker}*{Sec.5.4},
\begin{quotation}\em
(...) Similarities with the development in the preceding paragraph,
however, appear to stop here, as we see as unlikely the existence of
complete problems in either \NPiterU or $\NPiterB$.
This is so because
the property characterizing problems in $\Piter$ -- the fact that any
computation with a given \umach ‘measures its error’ to ensure that
	outputs in $\{\text{Yes}, \text{No}\}$ are correct -- does not appear to be checkable in
\Piter . That is, we do not know how to check in \Piter that, given
	a circuit $C$ outputting values in $\{\text{Yes}, \text{No}, \text{Unsure}\}$ a point $x \in \mathbb R^n$,
and a real \umach, all \umach-evaluations of $C$ with input $x$ return
the same value (...)
\end{quotation}

Our definition of classes \fpP and \fpNP will ensure that approximate computations can be certified. This will allow for the existence of $\NP$-complete problems.

\subsection*{Acknowledgements} We would like to thank Felipe Cucker and Marc Braverman for conversation and insight. \changedA{Two anonymous referees
significantly helped us to 
clarify some results and to improve the presentation of this paper.
}

\part{Exact computations over $\mathbb R$}

We first extend the BSS model of complexity over the reals to a model
where condition numbers are part of the input size. 
Computations are assumed {\em exact}.
Yet, this model is rich enough to fulfill most of our wish list.
For instance, it gives an answer to the objection by Braverman and Yampolsky.
In the model we present now, ill-conditioned 
instances are deemed to have
a large input size. Ill-posed instances have {\em infinite} input size.
A consequence is that some BSS-undecidable problems can be decided
in polynomial time with respect to their new input size.

In the example of the Julia set, points on its boundary
are ill-posed and their input size is infinite. The same
is true in the example of the graph of the exponential.
Therefore a machine can still be {\em polynomial time}
without accepting or rejecting ill-posed points. The machine can
work longer for points close to the Julia set or the graph of the
exponential, because those points are ill-conditioned.

We follow \ocite{Cucker} in considering the choice of 
the condition number function as a part of the problem. Sensible choices
include reciprocal distances to an ill-posed set. The metric used
to measure that distance makes a difference, so we do not prescribe
any restriction to the condition number at this stage.

\begin{definition} A {\em decision problem} is a pair $(X, \mu)$ 
where $X \subseteq \mathbb R^{\infty}$ and 
\changedA{$\mu: \mathbb R^{\infty} \rightarrow [1, \infty]$}
is an arbitrary function.
\end{definition}
\remark{In most of our work the reciprocal 
condition number $\mu^{-1}$ 
is Lipschitz and can be efficiently estimated a posteriori.}

\section{BSS machines}

In this paper, a {\em machine} is always a BSS machine over \changed{a
ring as in Blum et al. \ycite{BCSS}. More precisely, 
{\em BSS machines over $\mathbb R$} are machines
over the field of real numbers with division and {\em BSS machines
over ${\mathbf F}_2$} are machines over the finite field with two
elements. The complexity theory over ${\mathbb F}_2$ is known to be
equivalent to Turing complexity with respect to polynomial time.}
\changed{Following \ocite{BCSS},
$\mathbb R^{\infty}$ is the disjoint union of all the $\mathbb R^k$,
$k \in \mathbb N$ and $\mathbb R_{\infty}$ is the class of bi-infinite
sequences $(x_k)_{k \in \mathbb Z} \in \mathbb R$ with $x_k=0$ but for a finite number of $k$. The input and output
spaces of the machine are $\mathbb R^{\infty}$ but the state space is
$\mathbb R_{\infty}$.}
The machine is
assumed to be in a particular {\em canonical form}. 
This means that each node performs at most one arithmetical
operation. 
For later reference, we formally define:

\begin{definition}\label{BSS} A {\em machine} $M$ is a BSS machine over the
reals in {\em canonical form}:
\begin{enumerate}[(a)]
	\item The input node maps input 
$x=(x_1, \dots, x_L)$ into state 
\[
	s=I(x)= ( \dots, 0, 
		\underbrace{1, \dots, 1}_{L \text{\ times}}
		, 0. x_1, \dots, x_L, 0, \dots)
\]
where the dot is placed at the right of the zero-th coordinate in $\mathbb R_{\infty}$.
\item The output node maps state $s$ into output $y=O(s)$,
with $y_i=s_i$, $i \ge 0$.
\item All branching tests are of the form $s_0 > 0$.
\item The map 
$\mathbb R_{\infty} \rightarrow \mathbb R_{\infty}$
associated to a computation node is one of the following:
\begin{eqnarray*}
\opload{c}: &s \mapsto& ( \cdots, s_{-1}, c . s_{1}, \cdots )\\
\opadd{j}{k}: &s \mapsto& ( \cdots, s_{-1}, s_j + s_k . s_{1}, \cdots )\\
\opsub{j}{k}: &s \mapsto& ( \cdots, s_{-1}, s_j - s_k . s_{1}, \cdots )\\
\opmult{j}{k}: &s \mapsto& ( \cdots, s_{-1}, s_j s_k . s_{1}, \cdots )\\
\opdiv{j}{k}: &s \mapsto& ( \cdots, s_{-1}, s_j/ s_k . s_{1}, \cdots )\\
\opcpy{j}: &s \mapsto& ( \cdots, s_{-1}, s_j . s_{1}, \cdots )\\
\end{eqnarray*}
\item Constants $c$ are assumed to be real numbers.
\item \changedA{To each `fifth-node' is associated a map 
$\shiftl(s)$ or $\shiftr(s)$: $\mathbb R_{\infty} \rightarrow \mathbb R_{\infty}$ where
$\shiftl(s)_i= s_{i+1}$ and 
$\shiftr(s)_i= s_{i-1}.$ }
\item There is only one output node.
\item Each division is preceded by tests $s_0 > 0$ and $s_0<0$. In case
	both tests fail, $s_0$ does not change and the next node is the actual node itself, so the machine never terminates.
\end{enumerate}
\end{definition}

\changedA{Every BSS machine can be replaced by a machine in canonical form, at a cost of a linear increase in the number of nodes and in the running time.}
A machine $M$ \changedA{in canonical form} can be described by the number $N$ 
of nodes, and by a list of
maps associated to each node. Let
$j,k: \{1,\dots, N\} \rightarrow \mathbb Z$, $\beta^+, \beta^- :\{1, \dots, N\} \rightarrow \{2, \dots, N\}$, and $c: \{1, \dots, N\} \rightarrow$ 
$\mathbb R$. The letter $\nu$ denotes
the current node number and the table below gives the 
next-node map 
$\beta: 
\{1, \dots, N\} \times \mathbb R_{\infty}
\rightarrow
\{2, \dots, N\} 
$ and the next-state map $g:
\{1, \dots, N\} \times \mathbb R_{\infty}
\rightarrow
\mathbb R_{\infty}
$. Below, 
$I: \mathbb R^{\infty} \rightarrow \mathbb R_{\infty}$
and
$O: \mathbb R_{\infty} \rightarrow \mathbb R^{\infty}$
denote the input and output maps.

\centerline{
\begin{tabular}{|c|c|}
\hline\hline
Node type of $\nu$. & Associated maps
\\\hline\hline
Input node 
&
$\beta(1,s) = \beta^+(1) = \beta^-(1)$ \\
($\nu=1$)
&
$g(1,s)=I(x)$
\\ \hline
Computation node
&
$\beta(\nu,s) = \beta^+(\nu) = \beta^-(\nu)$
\\
$\nu \in \{2, \dots, N-1\}$
&
$g(\nu,s) \in \left\{ \opload{c(\nu)}(s), 
\opmult{j(\nu)}{k(\nu)}(s),
\opdiv{j(\nu)}{k(\nu)}(s),
\right.$\\
& \hspace{\stretch{1}}
$\left.\opadd{j(\nu)}{k(\nu)}(s), \opsub{j(\nu)}{k(\nu)}(s),
\opcpy{j(\nu)}(s) \right\}$.
\\\hline
Branch node
& $\beta(\nu,s) = \left\{
\begin{array}{ll}
\beta^+(\nu) &\text{if $s_0 > 0$} \\
\beta^-(\nu) &\text{if $s_0 \le 0$} 
\end{array}
\right.
$
\\
$\nu \in \{2, \dots, N-1\}$
& $g(\nu,s) = s$
\\\hline
Output nodes 
&
$\beta(\nu,s) = \beta^+(\nu)= \beta^-(\nu) = \nu$
\\
($\nu=N$). 
&
$g(\nu,s) = s$
\\
&
Outputs $(s_1, \dots, s_m)$.
\\\hline
Fifth node
& $\beta(\nu,s) = \beta^+(\nu) = \beta^-(\nu)$
\\
$\nu \in \{2, \dots, N-1\}$
& 
$ g(\nu,s) \in \{ \shiftl(s), \shiftr(s) \}$ where
\\ & $\shiftl(s)_i= s_{i+1}$, 
$\shiftr(s)_i= s_{i-1}.$
\\\hline\hline
\end{tabular}
}

\medskip
\par
\begin{definition}
An 
 {\em exact computation} is a sequence $((\nu(t), s(t))_{t \in \mathbb N_0}$ in $\{1,\dots, N\} \times \mathbb R_{\infty}$ 
satisfying $\nu(0)=1$, $s(0)=I(x)$ and for all $t \ge 0$,
\begin{eqnarray*}
\nu(t+1) &=& \beta(\nu(t), s(t)) \hspace{1em}
\\
s(t+1) &=& g(\nu(t), s(t)) 
\end{eqnarray*}
The computation {\em terminates} if $\nu(t) = N$ eventually, and the 
execution time $T=T_M(x)$ is the smallest of such $t$. 
The terminating computation is said to {\em accept} 
	input $x$ if $s_1(T)>0$ and to {\em reject} $x$ otherwise. 
The {\em input-output map} is the map $M: x \mapsto M(x)=O( s ( T(x) ))$ and 
the {\em halting set} of $M$ is the 
domain of definition of the input-output map.
\end{definition}

\begin{remark} For input $x=(x_1, \dots, x_n)$ of length $n$,
$i > n + t$ implies $s_i(t)=0$. Also, $i<-t$ implies
$s_i(t)=0$.
\end{remark}

A Universal Machine $U$ over $\mathbb R$ (actually over any ring)
was constructed
by~\ocite{BSS}*{Sec.8}.
Any machine $M$ over $\mathbb R$ can be described by
a `program' $f_M \in \mathbb R^{\infty}$ so that $U(f_M,x) = M(x)$.
This holds whenever either of the two machines stops.
Moreover if any of these computations finishes in finite time,
the running time
$T_U(f_M,x)$ for the universal machine
is polynomially bounded in terms of the running time 
$T_M(x)$ for the original machine. 

The result below will be used later. It is a trivial modification
of the argument proving the existance of the Universal Machine. Let $\Omega_{M,T}=\{x \in \mathbb R^{\infty}:\nu_T = N\}$
denote the time-$T$ halting set
of a machine $M$.
\begin{proposition}\label{timed}
There is a Machine $U'$ over $\mathbb R$ \changedA{with real constants
and} with the following
\changedA{property}: for any machine $M$ over $\mathbb R$,
for any input 
$x \in \mathbb R^{\infty}$,
\[
U'(T,f_M,x) = 
\left\{
\begin{array}{ll}
M(x) & \text{if $x \in \Omega_{M,T}$,}\\
	0   & \text{(reject) otherwise.}
\end{array}\right.
\]
Moreover, the running time
$T_{U'}(T,f_M,x)$ is polynomially bounded in terms of the running time 
$\min(T,T_M(x))$. 
\end{proposition}

\section{Polynomial time}

Recall that the input-output map for a machine $M$ with input $x$ is
denoted by $M(x)$ and the running time (number of steps) with input $x$
is denoted by $T_M(x)$. 
The {\em length} of an input $x=(x_1, \dots, x_n) \in \mathbb R^{\infty}$
is $\Length{x}=n$. 
We define the size of an instance $x$ of $(X,\mu)$ by
$\Size{x}= \Length{x} (1+\log_2 \mu(x))$.
In this paper, we make the convention that an output
$z=M(x) >0$ means YES 
and an output $z\le 0$ means NO.

In Turing and BSS complexity, the input size is its length which is
known. Therefore the
two following definitions of the class $\P$ of polynomial time decision problems are equivalent:
\begin{description}
	\item [`\changed{One-sided}' \P] $X \in \P$ iff there is a polynomial $p$ and a machine $M$ such
that for any $x \in X$ the machine $M$ with input $x$
		halts in time at most $p(\Length{x})$ and outputs a positive number, and for $x \not \in X$ the
machine does not halt.
\item [`\changed{Two-sided}' \P] 
$X \in \P$ iff there is a polynomial $p$ and a machine $\tilde M$ such
that for any $x$, the machine $\tilde M$ with input $x$
halts in time at most $p(\Length{x})$ and the output satisfies
\[
\tilde M(x) > 0 \Leftrightarrow x \in X.
\]
\end{description}
The delicate part of the argument for proving equivalence is the 
construction of the machine
$\tilde M$ given the machine $M$. This is done 
\changedA{in Proposition~\ref{timed}} by introducing a `timer' and
halting with a NO (negative) answer when the time is 
larger than $p(\Length{x})$.

Our model is different because the input size depends on
the condition, which is not assumed to 
be known.  \ocite{Braverman-Yampolsky} already explored the
idea of \changed{`one-sided'}\label{ref03} \P (actually \co\P) in their computer model, see
Example~\ref{BYweak}. They assumed the condition of {\em accepted} inputs can
be bounded conveniently. We do not make that explicit convention but all of our
examples admit an estimator.  We start with the \changed{one-sided} definition of $\P$.

\begin{definition}[Deterministic polynomial time]\label{P}
	The class \realPone (reads {\em \changed{one-sided}}~$\P$) of problems 
{\em recognizable in
polynomial time} is the
set of all pairs 
$(X, \mu)$ so that  
there is a BSS machine $M$ over $\mathbb R$ with input $x$, output
in $\mathbb R$ and with the following properties:
\begin{enumerate}[(a)]
\item There is a polynomial $p_{\Arith}$ such that {\bf whenever 
$x \in X$} and $\Size{x} < \infty$,
\[
T_M(x) < p_{\Arith}(\Size{x}).
\]
\item If $T_M(x) < \infty$, then 
\[
M(x) > 0 \Longleftrightarrow x \in X.
\]
\end{enumerate}
\end{definition}
The last condition implies in particular that for finite or infinite input size if an answer is given, it is correct.
\begin{definition}
	The class \realPtwo (reads {\em \changed{two-sided} \P}) of problems  
{\em decidable in polynomial time} is the
set of all pairs 
$(X, \mu)$ so that $(X, \mu)$ and $(\mathbb R^{\infty} \setminus X, \mu)$
are both in \realPone. 
\end{definition}

In this sense,
\[
\realPtwo = \realPone \cap \co\realPone
\]
Equivalently, one can remove the clause `{\em whenever $x \in X$}' 
from Definition~\ref{P}(a). 
Notice also that a
problem $(X,\mu)$ in $\realPone$ (resp. $\co\realPone$) 
`projects' in
$\realP''$ setting $\mu(x)=\infty$ for $x \not \in X$ (resp. $x \in X$).

\begin{example}\label{ex1} Let $X \in \mathbb R^{\infty}$. 
$X \in \PR$ if and only if $(X,1) \in \realPtwo$.
\end{example}

\begin{example} For any $X \subseteq \mathbb R^{\infty}$,
$(X,\infty) \in \realPtwo$.
\end{example}

The examples above are trivial. Below is a more instructive one.
It is known that $\mathbb Z \not \in \PR$ so that $(\mathbb Z,1) \not \in \realP''$.
If we plug in the correct condition number, then the `bits' of $\lfloor |x|
\rfloor$ can be found in polynomial time for all input $x$.

\begin{example}\changedA{Let $\mu(x)=1+|x|$. Then, $(\mathbb Z, \mu)$}
$\in \realP''$.
\end{example}

\begin{proof}
We consider the machine $M$ described by the following pseudo-code:

\begin{trivlist}
\item {\tt Input} $x$.
\item {\bf If} $x<0$ {\bf then} $x \leftarrow -x$.
\item $y\leftarrow1$.
\item {\bf While} $x \ge y$, 
\subitem $y \leftarrow 2y$.
\item {\bf While} $y \ge 2$, 
\subitem $y \leftarrow y/2$  
\subitem {\bf If} $x \ge y$ {\bf then} $x \leftarrow x-y$.
\item {\bf If} $x=0$ {\bf then} {\tt output} $1$ {\bf else} 
{\tt output} $-1$.
\end{trivlist}
When $|x|\ge 1$ each of the {\bf while} loops will be executed at most 
	$1+\log_2(\lfloor |x| \rfloor)$ \changedA{$\le \log_2(\mu(x))$} times. 
At the end of the second {\bf while}, $0 \le x < y$, and
$x - \lfloor x \rfloor$ is not changed after the first {\bf if}.
If $|x|<1$ the loops will not be executed at all, and the only
\changedA{possibly accepted} input is $x=0$.
\end{proof}

\begin{figure}
\centerline{\resizebox{7cm}{!}{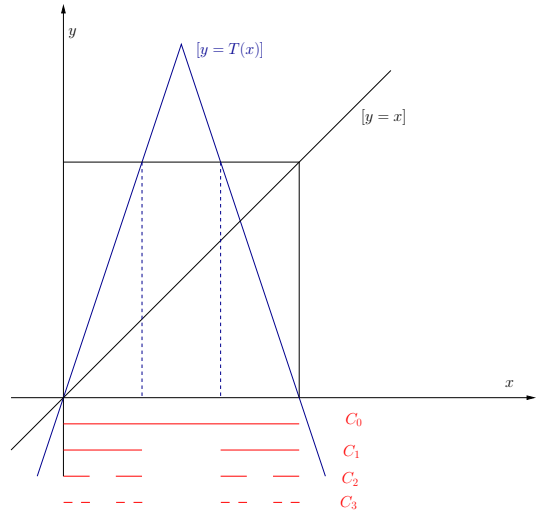}}
\caption{Construction of the middle-thirds Cantor set.\label{cantor}}
\end{figure}

\begin{example}\label{Cantor} Let $C$ be the {\em Cantor middle thirds set} (Fig.\ref{cantor}), 
\[
C = \left\{ 
\sum_{k=1}^{\infty} \frac{2 a_k}{3^k} : a_k \in \{0,1\} \right\} .
\]
It can also be constructed as 
$C = \cap_{k \ge 0} C_k$, with $C_0=[0,1]$,
$C_k=T^{-1} C_{k-1}$ 
where the {\em tent map} $T: \mathbb R
\rightarrow \mathbb R$ is
\[
T(x) = \left\{ 
\begin{array}{ll}
3 x & \text{for $x \le \frac{1}{2}$} \\
3-3 x & \text{otherwise.}
\end{array}
\right.
\]
Clearly, $C$ is not a countable union\label{ref04} of disjoint
points and intervals and therefore 
it is not 
BSS computable. Membership to $\mathbb R\setminus C$ can be
verified by iterating the tent map.
The condition number for the Cantor set is defined as
\[
\mu_C(x) = \frac{1}{\min( d(x,C), 1)} 
\]
where $d(x,y)=|x-y|$ is the usual distance. Then $\mu_C$ is
infinite in $C$ and finite in $\mathbb R \setminus C$. Moreover,
for $x \in C_{k-1}\setminus C_k$, we have always 
$d(x, C) \le 3^{-k}/2$ so $\mu_C(x) \ge 2 \times 3^k \ge 1$
and
\[
\Size{x} \ge 2+k \log_2(3)
.
\]
Since $k$ iterates are sufficient to check that $x \not \in C_{k}$,
it follows that 
$(\mathbb R \setminus C, \mu_C) \in \realPone$. If $x \in C$ then
$\mu_C(x) = d(x,C)^{-1} = \infty$ so we also have $C \in \realPtwo$.
\end{example}

A small modification of the example
gives us a sharp separation:
\begin{proposition}
\[
\realPone \ne \realPtwo .
\] 
\end{proposition}

\begin{proof}
Define $\mu(x) = \mu_C(x)$ for $x \not \in C$ and $\mu(x)=1$ otherwise.
We claim that $(\mathbb R \setminus C, \mu) \not \in \realPtwo$.
Otherwise, the decision machine would be supposed to recognize $x \in C$
in constant time. This is impossible since 
\changed{$C$ is not a countable union of points and intervals.}
\end{proof}

\begin{figure}
\centerline{\resizebox{\textwidth}{!}{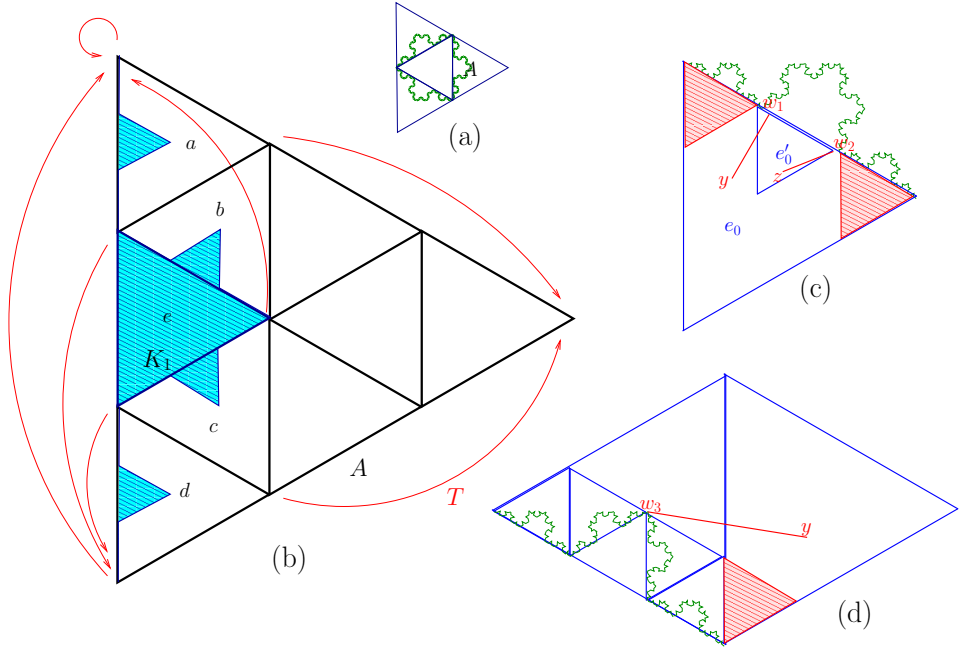}}
\caption{Koch's snowflake: (a) General view, (b) construction: the 
mapping $T$, (c) approximation of the distance from an interior point to the
border: for $y \in e_0 \subset e$, the approximation is the distance from
$y$ to the segment $e \cap b$. For $z \in e_0'$, the distance can be
approximated by the distance to $w_1$ or $w_2$. For points on the red
	\changed{triangles}\label{ref05}, use self-similarity. (d) The distance from
an exterior point to the border can be approximated by the distance to
	$w_3$ except for the colored triangle. For points in the \changed{triangle},
use self-similarity.\label{koch}}
\end{figure}

\begin{example} The {\em Koch snowflake} from Fig.~\ref{koch}(a)
can be treated in a
similar way. To simplify the presentation we will check
computability of the region delimited by a Koch curve inside
an equilateral triangle. Namely, let $\omega = e^{2 \pi i/3}$
and let $A$ be the solid triangle $(1, \omega, \omega^2)$, that is
the convex hull of $1$, $\omega$, $\omega^2$. Subdivide 
$A$ as in Fig.\ref{koch}(b) and consider a piecewise linear map
$T:A \rightarrow A$ which is continuous in each subdivision, and
maps each of $a$, $b$, $c$, $d$ as in the picture, is undefined 
in $e$ and in the remaining regions.
Then define $K_0 = e$ and inductively, 
$K_{t+1} = T^{-1}(K_t)$.
The piece of snowflake is $K = \cup K_t$.

We set $\mu_K(x) = d(x, \partial K)^{-1}$. Since $T_{|a \cup b \cup c \cup d}$
multiplies distances by 3, it takes at most 
$\lfloor \log_3 \mu_K(x) \rfloor$ iterations of $T$ to decide if 
$x \in K$.
This condition number is geometrically appealing. Another important property
that we do not require in our model (but see Ex.~\ref{BYweak})
is the capacity of estimating
the condition number {\em a posteriori}. In this example,
estimating the distance to $\partial K$ is easy. First assume that
$x \in K_{t+1} \setminus K_t$. Let $y = T^t(x) \in e$ so that
$d(x, \partial K) = 3^{-t} d(y, \partial K)$. Assume without
loss of generality that the imaginary part of $y$ is non-negative.
Then approximate the distance as in Fig.\ref{koch}(c). For points
outside $K$, iterate until $T^k(y)$ leaves $a \cup b \cup c \cup d$
and then estimate the distance as in Fig.\ref{koch}(d).
\end{example}

\begin{example}
The epigraph of the exponential, with condition number $\mu(x,y)=
	\changed{\max(|x|,1)/|e^x-y|}$, 
	is in $\realP''$.
The supporting algorithm for this problem was described by
~\ocite{Brent} in the context of floating point computations: the
	cost of computing $e^x$ \changed{
in a given interval,
	say $0 \le x \le 2$,} with absolute accuracy $2^{-n}$ 
is $O(M(n)\log (n))$ where $M(n) \ge n$ is the cost of multiplication
(Th.6.1). Since we are using a model that allows for exact computations,
the cost of the very same algorithm becomes $O(n \log(n))$. However, this
	bound is valid only for $x$ \changed{in the interval}. 
\par
In order to extend the result to the reals, assume first that
	$x > 0$ and write \changed{$x = x_0 2^a$} with $a \in \mathbb N$ and 
$1 \le x_0 \le 2$ \changed{or $a=0$ with $0 \le x \le 2$}.
Then $e^x = (e^{x_0})^{2^a}$. This means that we should compute $e^{x_0}$ 
with accuracy $2^{-O(n+2^a)}$ at a cost of $O( (n+|x|) \log(n+|x|))$. 
	For $x<0$, we just \changed{compare $e^{-x}$ to $1/y$ as above (no extra accuracy is needed
for the inverse)}
\changed{The following pseudo-code summarizes the procedure:
\begin{trivlist}\label{ref06}
\item {\tt Input} $x,y$.
\item {\bf If} $x\ge 0$,
\subitem {\bf then} $s\leftarrow 1$
\subitem {\bf else} $s\leftarrow -1$, $x \leftarrow -x$, $y \leftarrow 1/y$.
\item $a \leftarrow 0$, $x_0 \leftarrow x$.
\item {\bf While} $x_0>2$, {\bf do} $x_0 \leftarrow x_0/2$, $a \leftarrow a+1$.
\item {\bf $n_0 \leftarrow 8$}
\item {\bf Repeat}
	\subitem Apply Brent's algorithm to compute $z_0 \leftarrow e^{x_0}$ with accuracy $2^{-n_0}$.
	\subitem Compute $z \leftarrow z_0^{2^a}$ 
	by repeated squaring.
	\subitem Compute $E \leftarrow (1+2^{-n_0})^{2^a}$ 
	by repeated squaring.
	\subitem {\bf If} $y > z E$ {\bf then} {\tt output} $s$.
	\subitem {\bf If} $y < z / E$ {\bf then} {\tt output} $-s$.
	\subitem $n_0 \leftarrow n_0+1$.
\end{trivlist}
}
\end{example}

\begin{example}\label{BYweak}
 A notion of {\em weakly computable set} of $\mathbb R^n$
	was explored by ~\ocite{Braverman-Yampolsky}.\label{ref07}
In their model a point $x \in \mathbb R^n$ is represented by an
`oracle' function that given $m$, produces a $2^{-m}$-approximation
with cost $O(m)$.
A set $K \subset \mathbb R^n$ is (weakly) computable if there is an
oracle Turing machine $T_K$ that given
	a \changed{point} $x \in \mathbb R^n$ (represented by an oracle)
and given $m$,
\begin{itemize}
\item answers 1 (true) if $x \in K$,
\item answers 0 (false) if $d(x,K) > 2^{-m+1}$,
\item answers 0 or 1 otherwise.
\end{itemize}
Above, $d$ is the Euclidean distance. If
the answer is $0$ we can infer that
	$x \not \in K$. An answer of $1$ is inconclusive.
\par
A set can be computed in polynomial
time if and only if the machine $T_K$ terminates in time polynomial in $m$.
Examples of computable polynomial time Julia sets are given in their book
	(for instance, Th 3.4 \changed{p.42} ).
\par
Given any closed, bounded and 
weakly computable set $K$, we may define the condition number
$\mu(x) = 2^m$ for $x \not \in K$, where $m$ is minimal so that $T_K$ 
	with input \changed{$(x,m)$} returns $0$. Of course, $\mu(x) = \infty$ for $x \in K$.
	Then $(\changed{\mathbb R^n \setminus K}, \mu) \in \realPone$. Indeed, given any
$x$, it is easy in the BSS model to obtain diadic approximations 
of $x$ in polynomial time. This replaces the oracle. Then we can simulate
the Turing machine $T_K$ within the same time bound. 
\par
Because we assumed that $K$ is closed, $x \not \in K$ implies that 
	$d(x, \partial K)>0$. Taking $m > -\log d(x, \partial K) +1$ \changed{already} guarantees that
the machine $T_K$ with input $x$ answers $0$. 
\end{example}

\begin{remark}
A more natural definition for the example above would be
$\mu(x) = d(x, \partial K)^{-1}$, and we would have $m < \log \mu(x) + 1$
anyway. This condition number can possibly be much larger than the original
one. 
\end{remark}

\begin{example}[canonical condition]
This is a generalization of the condition number of 
Example~\ref{BYweak}. 
Let $M$ be an arbitrary BSS machine over $\mathbb R$.
Let $\Omega_{M,T}=\{x: T_M(x) = T\} 
\subset \mathbb R^{\infty}$ 
be the time-$T$ halting set of $M$. 
Let $\Omega_M = \cup_{T \ge 1} \Omega_{M,T}$ be
the halting set of $M$. Then, define
\[
X = \{ x \in \Omega_M: \text{$M$ accepts $x$}\}
\hspace{1em}
\text{and}
\hspace{1em}
\mu(x) = 
\left\{
\begin{array}{rl} 
2^T & \text{if $x \in \Omega_{M,T}$ } \\
\infty & \text{if $x \not \in \Omega_M$.}
\end{array}
\right.
\]
Then $(X, \mu)$ is in $\realP''$. We will call $\mu$ the {\em canonical
condition} associated to a machine $M$. Given any machine $M$
accepting $X$,
one always have $(X, \mu) \in \realPtwo$. 
This is exactly the same trick as
increasing the size of an input instance in discrete computability theory.
\end{example}
The previous examples of problems $(X,\mu) \in \realPone$ 
admit an estimator for $\log_2 (\mu(x))$ when $x \in X$ up to
a bounded relative error.

More generally, for machines not necessarily solving a \changed{decision}\label{ref08} 
problem, a function $f(x)$
can be
estimated in polynomial time (with respect to some input size) if
and only if given $\epsilon$, there is a machine that produces
an $\epsilon$-approximation of $f(t)$ in time polynomial in the
input size and $\log \epsilon^{-1}$. For instance, the $d$-th root
of $t$ for $t>1$ can be approximated in polynomial time.
An easy modification of our previous example yields:

\begin{theorem}
Let $(X,\mu) \in \realPone$. Then there is $\mu' \le \mu$ so
that $(X, \mu') \in \realPone$ and $\log_2 \mu'$ 
can be estimated in
polynomial time with respect to the input size for $x \in X$.
\end{theorem}
\begin{proof}
	Let $M$ be the machine in Definition~\ref{P}. There \changedA{are $c,d>0$} so
that for $x \in X$,
\[
	T_M(x) \le \changedA{c} \Size{x}^d = \changedA{c} \Length{x}^d (1+\log_2 \mu(x))^d
\]
We define $\mu'(x)$ by solving the equation
\[
	T_M(x) = \changedA{c} \Length{x}^d (1+\log_2 \mu'(x))^d
\]
that is
\[
	\mu'(x) = 2^{\frac{ \sqrt[d]{T_M(x)\changedA{/c} } }{\Length{x}}-1}
\]
so that $\mu'(x) \le \mu(x)$ and $(X,\mu') \in \realPone$ \changedA{is} decided
by the same machine $M$. 
\end{proof}

\section{Non-deterministic polynomial time}

\begin{definition}[Non-deterministic polynomial time]\label{NP}
The class \realNPone of problems recognizable in non-deterministic
polynomial time is the
set of all pairs 
$(X, \mu)$ so that  
there is a BSS machine $M$ over $\mathbb R$ with input $(x,y)$, output
in $\mathbb R$ and with the following properties:
\begin{enumerate}[(a)]
\item There is a polynomial $p$ such that whenever 
$\Size{x} < \infty$ and $x \in X$, there is $y$ such that
		$M(x,y)>0$ and
\[
T_M(x,y) < p(\Size{x}).
\]
\item If $T_M(x,y) < \infty$ and $M(x,y)>0$, then $x \in X$.
\end{enumerate}
\end{definition}

The possibility of rejecting an unlucky guess $y$ for $x \in X$
is irrelevant and we
can replace the machine in the definition by a machine that can 
\changedA{either work forever
or accept the input.}
Clearly, $\realPtwo \subseteq \realPone \subseteq \realNPone$. 
A few trivial examples in \realNPone are:

\begin{example} Let $X \in \mathbb R^{\infty}$. $X \in \realNP$ if
and only if $(X,1) \in \realNPone$.
\end{example}

\begin{example} For any $X \subseteq \mathbb R^{\infty}$,
Then $(X,\infty) \in \realNPone$.
\end{example}

The following problem is in $\realNP$: 
\begin{example} Let \SAFeas\ (Semi-Algebraic Feasibility) be the set of all $(n,f)$
where $n \in \mathbb N$ and $f$ is a system of 
real polynomial equations in $n$ variables codified in sparse representation,
such that there is some $y \in \mathbb R^n$
with $f(y)>0$ (coordinatewise). The problem $\SAFeas$ belongs to $\realNP$
because
a non-deterministic
machine $M=M(f,y)$ can guess $y$ and compute $f(y)$ \cite{BCSS}*{Prop.3 p.103}. 
In particular $(\SAFeas,1) \in \realNPone$. 
\end{example}

\begin{example} A modification of the previous example:
Let \SAEFeas\ (Semi-Algebraic-Exponential Feasibility)
be the set of all $(n,f)$
where $n \in \mathbb N$ and $f$ codifies a system of 
real polynomial equations in $2n$ variables such that there is some 
$y \in [0,1]^{n}$
with $f(y,e^y)>0$ (coordinatewise). The machine $M=M(f,y,\epsilon)$
will compute each $e^{y_i}$ approximately up to relative error $\epsilon$.
There will be for each $f$ in the yes set a best guess $y$ and a value of
$\epsilon$ such that $\epsilon$-approximations of $e^{y_i}$ are sufficient
to infer that
$f(y, e^y)>0$. The condition number for such $f$ will be deemed to be 
$1/\epsilon$.
\end{example}

\begin{example}\label{ref09} This is a basic geometric example of problem in $\realNPone$.
\changed{For definition and references on the geometric concepts,
we recommend the textbook by \ocite{Berger}}.
	Let $S \subseteq \mathbb R^n$ be a smooth algebraic variety with \changed{maximal} 
principal curvature $\le \kappa_{\max}$. 
Let $d$ denote the
Euclidean distance in $\mathbb R^n$, while $d_{\ell}$ is the
distance along $S$, 
\[
d_{\ell}(x,y) = 
\changed{
\lim_{\epsilon \rightarrow 0} 
\inf_{\substack{
x_0=x,\  
x_N=y\\
x_i \in S,\  i=1 \dots N-1 \\
d(x_i, x_{i+1})<\epsilon,\  i=0 \dots N-1\\
}}}
\sum_{i=0}^{N-1} d(x_i, x_{i+1})
.
\]
\changed{ 
Let $\delta_0$ be such that
$S$ admits a $\delta_0$-tubular neighborhood. 
This means that the neighborhood $\{x \subset \mathbb R^n: d(x, S) <
	\delta_0\}$ is diffeomorphic to the normal bundle of $S$.}

Let $x$ be a fixed point of $S$ and let $r>0$. 
Then define
\[
X = \{ y \in S: d_\ell(x,y) < r\}
\text{ and }
\mu(y) = 2^{ 1/|r - d_\ell(x,y)| }
.\]
We claim that $(X, \mu) \in \realNPone$. 
\end{example}
\begin{proof}
The supporting algorithm is as follows:
\begin{trivlist}
\item {\tt Input} $(y; x_1, \dots, x_{N-1}, \delta)$.
\item $x_0 \leftarrow x$, $x_N \leftarrow y$.
\item \changed{{\bf If} $y \not \in S$ {\bf or} some $x_i \not \in S$
	\bf{ then } {\tt output} -1. }
\item {\bf If} $\delta \ge \delta_0$ {\bf then} {\tt output} -1.
\item {\bf For} $0 \le i \le N-1$, 
\subitem {\bf If} $\|x_{i+1} - x_i\| > \delta$ {\bf then} {\tt output} -1.
\item $m \leftarrow r - \sum_{i=0}^{N-1} \|x_{i+1}-x_i\|$.
\item {\tt Output} $m$.
\end{trivlist}

Before proving correctness of the algorithm, we notice the following
fact: if $x(t): [0,T] \rightarrow S$ is a minimizing geodesic, then
\begin{equation}\label{geodesic-estimate}
T - \frac{\kappa_{\max}T^2}{2}
\le
\| x(T) - x(0) \|  
\le T
.
\end{equation}
	\changed{The upper bound above follows from the triangle inequality.
	To establish the lower bound on \eqref{geodesic-estimate},} we write
\[
x(T)-x(0) = \int_{0}^T \dot x(t) \dd t =  T \dot x(0) + \int_0^T \int_0^t
\ddot x(s) \dd s \dd t .
\]
Then we use $\|\dot x(0)\|=1$, $\|\ddot x(t)\| = \kappa(t) \le \kappa_{\max}$
and triangular inequality.
Now assume that $y \in X$. Pick 
	\begin{equation}\label{eq-N}
N = 
\max
\left(
\left \lceil
\frac{\kappa_{\max} r^2 \log_2 \mu(y)}{2}
\right \rceil
,
\left \lceil
\frac{r}{\delta_0}
\right \rceil
\right)
	.\end{equation}
and set $\delta = d_{\ell}(x,y)/N$.
	By the Hopf-Rinow theorem \changed{\cite{Berger}*{th.52}},
there is a minimizing geodesic $x(t)_{t \in [0,N\delta]}$ between
$x$ and $y$. Set $x_i=x(i \delta)$. With those choices, the
algorithm computes $m = r - \sum_{i=0}^N \| x_{i+1} - x_i \|$.
	\changed{Upon acceptation,} equation \eqref{geodesic-estimate} yields
\[
0< r - d_{\ell}(x,y) \le m \le r -d_{\ell}(x,y) + \frac{\kappa_{\max} d_{\ell}^2}{2N}
\le
	2(r - d_{\ell}(x,y) \changed{)}
\]
\changed{where the last inequality is a consequence of \eqref{eq-N} and
	of the choice of $\mu$.}
	\changedA{The following estimate for $\mu$ also follows:}
\[
\mu(y) \le 2^{2/m} \le \mu(y)^2
\]
\end{proof}
\changed{
\begin{remark}
In the example above, the subset $X$ is contained in the connected
component $S' \subset S$ containing $x$. 
By taking $r$ as an extra input to the supporting algorithm, 
one can also deduce that $(S',1)$ is in $\realNPone$.
\end{remark}
} 

\begin{remark} In many problems of interest, the reciprocal of the condition number is equal or related
to the {\em distance} to the set $\Sigma$ 
	of degenerate inputs. The choice of the metric depends usually on the problem one wants to solve, and there may be several workable choices. In the context of this paper, one can make the subset of inputs $x \in \mathcal I$ with finite or infinite condition $\mu(x)$ into a \changedA{pseudo}-metric space by defining $d(x,y)= | \mu(x)^{-1} - \mu(y)^{-1} |$. This setting has the inconvenience of attibuting distance zero to different problems with the same condition. Another possibility is setting
\[
d(x,y) = | \mu(x)^{-1} - \mu(y)^{-1} | + \|x-y\|_2
\]
which provides an inequalty $\mu(x)^{-1} \le d(x, \Sigma)$.
	\end{remark}
\begin{remark} 
A {\em path metric
space} is a metric space where the distance between two points is the
infimum of the length of the curves between those two points. For 
instance, Riemannian manifolds are path metric spaces. 

A necessary and sufficient condition for a complete metric space
$(S,d)$ to be a path metric space is that for arbitrary points
$x, y \in X$ and for each $\epsilon > 0$, there is $z$ such that
\[
\max( d(x,z), d(z,y) ) \le \frac{1}{2} d(x,y) + \epsilon
.
\]
\cite{Gromov}*{Th.1.8 p.7}.
It may be possible
to generalize the example above to other path metric spaces.
A common situation is to have a subset $S$ (e.g. a manifold)
embedded into another metric space (e.g. $(\mathbb R^n,d)$ or
$(\mathbb P^n,d)$). The subset $S$ inherits the metric of the ambient space,
but we can also define a path metric along $S$.

In general, $d \le d_{\ell}$ but it is hard to obtain upper bounds
for $d_{\ell}$ ~\cite{Gromov}*{Sec.1.15$\frac{1}{2}_+$}.
The example above seems easier to generalize when such upper bounds are
available.
\end{remark}

\section{\NP-completeness}

NP-hardness will be defined through \changed{one-sided} Turing reductions.
Informally, a {\em Turing reduction} from a problem $(X,\mu)$ to
a problem $(Y, \eta)$ is a BSS machine for $(X, \mu)$ that is also 
allowed to repeatedly query a machine for $(Y, \eta)$. This reduction is
said to be a polynomial time reduction if and only if, for all $x \in X$
the machine for $(X, \mu)$ runs in polynomial time, and produces polynomially
many queries to the machine for $(Y, \eta)$, and the size of each query is
polynomially bounded \changed{in} the input size $\Size{x}$.
This definition ensures that given a polynomial time 
Turing reduction,  
$(Y,\eta) \in \realPone$ implies $(X,\mu) \in \realPone$.

\begin{remark}[Many-one vs Turing reductions]
A stricter notion of reducibility was used by \ocite{BCSS}*{Sections 5.3, 5.4} through the use of {\em p-morphisms}. In their definition, the reduction machine can call the machine for $(Y,\eta)$ only once. This is also called a {\em many-one reduction}. The reduction of a general problem $X \in \realNP$ to the canonical $\realNP$-complete problem goes by a reduction to the register equations/inequations up to time $T$. 
This value of $T$ can be bounded polynomially in terms of the input size, which is known in the BSS model. In this paper the input size depends on the condition which is not assumed known. Hence the p-morphism based argument 
fails, and we need to consider Turing reductions instead.
\end{remark}

\begin{definition}\label{oracle1}
A BSS machine with a black box for $(Y, \eta)$
({\em black box machine} for short)
is a BSS machine $M$ over $\mathbb R$
with an extra node $\nu_O$, the {\em black box} node. It has 
one outgoing edge $(\nu_O,\beta(\nu_O))$. 
\end{definition}

The {\em black box} 
can be thought as a subroutine 
to compute
a certain arbitrary function
$O: \mathbb R^{\infty} \rightarrow \mathbb R$. This subroutine will
have to satisfy certain properties regarding correctness and time
cost. In modern 
programming this sort of routine is called an {\em abstract method} 
while in traditional computer science it is called an {\em oracle}.

When the black box node is attained (say at time $t$), 
it interprets a fixed set
of the state variables $s_t$ as an input of the form 
$(S,y)$ to the `subroutine'. If $y \in Y$ and $\Size{y}\le S$ the black box node
will replace $S$ with a positive output $O(y)$. 
If the black box produces a positive
output given $(S,y)$, then $y \in Y$. The black box always produces an answer,
and if  $\Size{y}>S$ it may `time out' and replace $S$ by a 
negative value. A more precise description can be given through the
register equations.

\changed{In classical complexity theory, black boxes are assumed to give
an answer in unit time. The total computation time of a black box machine
with input $x$
includes the time for preparing each black box input $y$, 
so $\Length{y}$ is still a lower bound for the running time. Whence,
$\Length{y}$ is polynomially bounded with respect to $\Length{x}$.
If the black
box is replaced by a routine that runs in time polynomial \changedA{in} $\Length{y}$,
the total running time will be polynomial \changedA{in} $\Length{x}$.
In this
paper, $\Size{y} = \Length{y}(1+\log_2 \eta(y))$ is not known at the time of the query, so the same trick is not available any more.
This forces us to depart from the classical complexity theory
and to assume instead that black boxes answer in a certain prescribed
time $S$. There are no false positives, but the black box may fail
to accept $y \in Y$ if 
$\Size{y} > S$.}

\begin{definition} \label{oracle2}
The register equations for a machine $M$ with a black box
for the problem
	$(Y,\eta)$\label{ref11} 
are the same as the register equations for $M$, except when the black box 
node $\nu_O$ is attained.
If $\nu_O$ is first attained at $t=t_0$ and the machine is at internal state
	$s(t_0) = \changed{(z, S. y)}$, then
		\[
(\nu(\tau), s(\tau)) = 
\left\{
\begin{array}{ll}
	(\nu_0, \changed{(z,S. y)}) 
&
	\text{for $0 \le \tau - t_0 < \changed{S}$;}
\\
	(\beta(\nu_0), \changed{(z, r. \dots)}) & \text{If $\tau-t_0=\changed{S}$ with 
}\\
\multicolumn{2}{c}{
	\changedA{r=\left\{\begin{array}{ll}
	+1 & \text{if $y \in Y$ and $\Size{y} \le S$,}\\
	-1 & \text{if $y \not \in Y$,}\\
	\pm 1 & \text{in all the other cases.} \end{array}\right.
	}}
\end{array}
\right.
\]
\end{definition}

The {\em computation} for a machine $M$ with a black box and a given input is just the output associated to the register equations. While the size of each query to the black box is not assumed to be known, one strategy to build such a machine is to keep doubling the bound $S$ until the input is hopefully accepted.
\begin{definition}\label{Turing-reduction}
	A polynomial time \changed{(one sided)} 
Turing reduction from $(X,\mu)$ to $(Y,\eta)$ is a BSS machine $R$
over $\mathbb R$ with a black box for $(Y,\eta)$ such
that any computation for $R$ with input $x$  
satisfies: 
\begin{enumerate}[(a)]
\item There is a polynomial $p$ such that whenever $x \in X$ and
$\Size{x} < \infty$,
$T_{R}(x) < p({\Size{x}})$.
\item If $T_{R}(x) < \infty$, then 
\[
R(x) > 0 \Longleftrightarrow x \in X
.
\]
\end{enumerate}
\end{definition}

\begin{lemma} If there is a polynomial time \changed{one-sided} 
Turing reduction from $(X,\mu)$ into
$(Y, \eta)$ and if $(Y, \eta) \in \realP'$, then $(X,\mu) \in \realP'$.
\end{lemma}

\begin{proof}
	Let $M_Y$ be the polynomial time machine of
	Definition~\ref{P} for $(Y, \eta)$. Let $q$ be the polynomial
time bound for $M_Y$. Then the black box node should be replaced by
The machine of Proposition ~\ref{timed} simulating $M_Y$ with input $y$, up to time
$q(S)$.  
	This simulation can be done in time polynomial
in $q(S)$. 
The composite machine satisfies
	conditions (a) and (b) of Definition ~\ref{P}. 
Therefore it satisfies the
requirements in Definition~\ref{P} for $(X, \mu)$.
\end{proof}

\begin{definition}
A problem $(Y, \eta)$ is $\realNPone$-hard if and only if
for any problem $(X, \mu) \in \realNPone$, there is a 
polynomial time Turing reduction
from $(X, \mu)$ to $(Y, \eta)$. A problem $(Y, \eta)$ is $\realNPone$-complete
if $(Y, \eta) \in \realNPone$ and $(Y, \eta)$ is $\realNPone$-hard.
\end{definition}

\begin{theorem}
If there is one $\realNPone$-complete problem in $\realPone$, then
$\realPone = \realNPone$. 
\end{theorem}

Here is a striking result: there is a $\realNPone$-complete problem
with constant condition. 
\begin{theorem}\label{SAFeas-complete} $(\SAFeas,1)$ is $\realNPone$-complete.
\end{theorem}
We will see in the proof that the 
length of the semi-algebraic
systems produced during the reduction is assumed to be polynomial
in the size of $x \in X$. Thus, 
\changedA{an input $x$ with a large condition number may give rise
to a long
semi-algebraic system, rather than to a system with large condition.}

\begin{proof}[Proof of Theorem~\ref{SAFeas-complete}]
We already know that $(\SAFeas,1) \in \realNPone$. 
Now, let $(X, \mu) \in \realNPone$ and
let $M$, $p$ be as in Definition~\ref{NP}. Without loss of generality
assume that $M(x,g)$ does not use coordinate $g_t$ before time $t$. 
Let  
$\Phi_T=\Phi_T(g_1, \dots, g_T)$
be the time-$T$ register equations 
for $M$ with input $(x,g_1, \dots, g_T)$, and add the extra requirements
that the computation terminated ($\nu(T)=N$)
and accepted the input ($s_0(T)>0$).
\par
Theorem 2(1) pp.78--79 in \cite{BCSS}
guarantees that
$\Phi_T$ is a system of at most $cT^2$ polynomial equations of 
degree $\le c$ in at most $2T+cT^2$ variables, plus at most
$2T$ linear inequalities. The constant $c$ depends only on $M$.
Therefore, the size of $\Phi_T$ is polynomially bounded in $T$.
Let $r \in \mathbb N$ be be large enough, so that
$\Size{\Phi_T} \le T^r$ for all
$T\ge 1$.
Consider now the following black box BSS machine:

\begin{trivlist}
\item {\bf Input} $x$.
\item $T \leftarrow \Length{x}$
\item {\bf Repeat}
\subitem $T \leftarrow 2T$.
\subitem {\tt Produce  
$\Phi_T=\Phi_T(x,g_1, \dots, g_T)$.}
\subitem {\bf Until} {\tt a query of the black box with 
input $(T^r,\Phi_T)$ succeeds.}
\item {\rm Return 1}.
\end{trivlist}

This machine will never accept $x \not \in X$. If $x \in X$
then $x$ will eventually be accepted. If $x \in X$ has finite
size, then eventually 
\[
p(\Size{x}) < T \le 2p(\Size{x}),
\]
so the size of the register equations is $\le T^r$. 
Therefore the black box machine above terminates for
$T < 2p(\Size{x})$. Each execution of the
loop takes at most $T^R$ steps for a certain $R \in \mathbb N$
so that the total
running time of the black box machine is bounded by
\[
2 (2p(\Size{x})^R.
\]
which is polynomial in the size of $x$.
\end{proof}

\begin{theorem}\label{toy-transfer} The following are equivalent:
\begin{enumerate}[(a)]
\item $\realP \ne \realNP$.
\item $\realPone \ne \realNPone$.
\end{enumerate}
\end{theorem}

\begin{proof}[Proof of Theorem~\ref{toy-transfer}]
Assume first that $\realP = \realNP$. 
We know from \ocite{BCSS}*{Prop.3 p.103} that
$\SAFeas \subset \realNP$ so
$\SAFeas \in \realP$. Then example~\ref{ex1} 
implies that  $(\SAFeas, 1) \in \realPtwo
\subsetneq \realPone$.
From Theorem~\ref{SAFeas-complete}, $(\SAFeas, 1)$ is
$\realNPone$-complete. So $\realPone=\realNPone$.

Reciprocally, assume $\realP'=\realNP'$. 
According to \ocite{BCSS}*{Th.1(2a) pp.104--105} 
	$\SAFeas$ \changedA{is} $\realNP$-complete. 
The length of any instance of $\SAFeas$ is known, so the
machine to decide $(\SAFeas,1)$ can be simulated by another machine
which keeps also track of time. This machine will always stop and produce an
answer, either the answer of the original machine or a NO in case of timeout.
This way it will always answer correctly to the question $x \in \SAFeas$.
The simulation machine will therefore decide
$\SAFeas$ in polynomial time, and therefore $\SAFeas \in \realP$.
\end{proof}

Moreover, deciding feasibility for a semi-algebraic set is known to be in
exponential time \cite{BPR}*{Theorem 13.14 p.475}. 
Historical references and an \changedA{earlier} algorithm can be found in 
~\cite{Renegar}\label{ref12}.
Thus,
\begin{theorem}
\[
\realNPone \subseteq \realEXPone.
\]
\end{theorem}

\begin{definition}
An {\em algebraic decision circuit} $C$ with input $x \in \mathbb R^n$ and
constants $y \in \mathbb R^m$ is a labeled directed acyclic
graph of order $\tau > m+n$ 
with nodes of indegree $0$, $2$, $3$ and 
	exactly one node of outdegree $0$ (the {\em output} node, labelled $\tau$). 
To node $i$ is associated a real variable $w_i$ and a formula
of one of the following types:
\begin{enumerate}[(a)]
\item For nodes of indegree $0$, 
\begin{eqnarray*}
w_i \leftarrow x_i&& \text{for $1 \le i \le n$,}\\
w_i \leftarrow y_{i-n}&& \text{for $n < i \le m$,}
\end{eqnarray*}
and no node for $i>m+n$ has indegree $0$.
\item For nodes of indegree 2 and direct predecessors $j=j(i)<i$ 
and $k=k(i)<i$, 
\[
w_i \leftarrow w_j \circ_i w_k \hspace{2em} \text{where $\circ_i \in \{+,-,\times, /\}$.}
\]
\item For nodes of indegree 3 ({\em selectors}) and predecessors $j=j(i)<i$, 
$k=k(i)<i$, $l=l(i)<i$,
the expression:
\[
w_i \leftarrow S(w_j,w_k,w_l) = \left\{
\begin{array}{ll}
w_j & \text{if $w_l > 0$,} \\
w_k & \text{otherwise.} \\
\end{array}
\right.
\]
\end{enumerate}
\end{definition}
An {\em exact computation} for the algebraic decision circuit $C$
with input $x$ is a sequence $w_1, \dots, w_{\tau}$ satisfying
$w_i=x_i$ for $1 \le i \le n$, $w_i=y_{i-n}$ for $n<i\le m$
and then $w_i=w_{j(i)} \circ_i w_{k(i)}$ or 
$w_i=S(w_{j(i)},w_{k(i)},w_{l(i)})$. 
The {\em size} of a circuit
is the number $\tau$ of nodes. 
We are assuming implicitly
that exact computation admit no division by zero, so let 
$\Omega$ be the set of all $x$ admitting an exact computation. 
Given an input $x \in \Omega$, we denote by
$C(x)$ the value of $w_{\tau}$.
 
\begin{example}Let \CircFeas\ be the set of circuits $C$ such that there
is an input $x \in \Omega$ with $C(x)>0$. 
\end{example}
\begin{theorem}\label{CircFeasNPc}  $(\CircFeas,1)$ is $\realNPone$-complete.
\end{theorem}

\begin{proof}
The proof that $(\CircFeas, 1)$ is in \realNP is to produce a circuit
simulator with input $(C,x,w)$. Then it will check the equations and
inequations for the associated straight line program. 

The proof of completeness loosely follows the argument by
\ocite{Cucker-Torrecillas}:
Let $M$ be a machine
as in definition~\ref{NP} with input $x=(y,z)$.
Recall from Definition~\ref{BSS} that $M$ is assumed to be in canonical
form. We will associate to $M$ a family of circuits.
Let $(\nu(t),s(t))$ be an exact computation for a machine $M$ with
input $x=(y,z)$. At time $t\le T$, only states 
$s_{-T}(t), \dots, s_{-1}(t),s_0(t),s_1(t),\dots,
s_T(t)$ are actually relevant.

It is possible to write
each of $\nu(t)$, $s_i(t)$ as the result of a circuit depending on the
$\changed{\nu(t-1)}$, $s_i(t-1)$. 
For instance, $\changed{\nu(t)}$ can be computed as follows. Both
$\beta^+(\nu)$ of $\beta^-(\nu)$ can be represented by
Lagrange's interpolating polynomials
\[
\beta^+(\nu) = \sum_{i=1}^N \prod_{j\ne i} \frac{\nu-i}{j-i} \beta^+(i)
\hspace{2em}
\text{and}
\hspace{2em}
\beta^-(\nu) = \sum_{i=1}^N \prod_{j\ne i} \frac{\nu-i}{j-i} \beta^-(i)
.
\]
The interpolating polynomials can be encoded as algebraic circuits 
of size $O(N^2)$. Then the value of $\nu(t)$ can be computed by a selector
with input $s_0(t-1)$, $\beta^+(\nu(t-1))$, $\beta^-(\nu(t-1))$.
Each state $s_k(t)$ can be written also as
\[
s_k(t) = \sum_{i=1}^N \prod_{j\ne i} \frac{\nu(t-1)-i}{j-i} g(\nu(t-1), s(t-1))
\]
where $g(\nu,s)$ is the associated map for node $\nu$ and state $s$.
	\changed{For each possible value of $\nu = \nu(t-1)$\label{ref14},
$s_k(t)=g_k(\nu, s)$ is a constant, or a monomial in one or two
variables $s_j=s_j(t-1)$.} Therefore $s_k(t)$ can also be computed by
a circuit of size polynomial in $T$.

Since there are $O(T^2)$ such circuits, the overall size of the
time-$T$ circuit is bounded by a polynomial in $T$. Thus, 
there is a polynomial
bound $q(T)$ on the time necessary to produce the time-$T$ circuit.
Now assume the notations of Definition~\ref{Turing-reduction}. Suppose
that $O$ is a black box deciding $(\CircFeas,1)$.
Then we set
\[
	p(T) = \sum_{t=1}^T \changed{t+q(t)}
\]
which is also polynomial in $T$.
\end{proof}

\changed{
\begin{remark}\label{ref13}
An alternative proof for the completeness of $(\CircFeas,1)$ can be
obtained by reduction from $(\SAFeas,1)$. 

\end{remark}
}

\part{Approximate computations}

\section{Floating point numbers}

Modern digital computers perform numerical calculations 
on a discrete, finite set of {\em floating point numbers} rather
than on general real numbers.
The current arithmetic standard is IEEE 754 \ycite{IEEE} by the Institute of Electrical
and Electronic Engineers. 
Real numbers are approximated
by floating point numbers of the form \changed{$m \beta^{e}$}\label{ref15}
where the mantissa $m$ and the exponent $e$ are integers 
with $|m|<\beta^{t+1}$, $t$ is the number
of digits and $\beta$ is the {\em radix}. 
We avoid in this paper certain technical details of the standard.
For instance under IEEE 754 the radix $\beta$ is supposed to be either 2 or 10, and
most digital computers use radix 2. The newer IEEE 854 standard
predicts arbitrary radix floating point numbers. For simplicity,
we will \changed{assume that $\beta=2$ for all the arithmetical lemmas.
Our complexity results remain valid
for every $\beta > 1$, but proofs are more elaborate}.
The IEEE 754
standard stipulates that $t$ should belong to a certain range. 
We will make no such assumption. As a consequence of this restriction on the exponent,
IEEE 754 allows for the representation
of {\em normal} and {\em subnormal} numbers. 
We assume that non-zero floating point numbers
are represented in {\em normal} form, that is $\beta^{t} \le |m|$.
Zero has a special representation. 
IEEE standards allow for special values like $\pm \infty$ and
NaN ({\em Not a Number}) to handle divisions by zero, but our simplified
model assumes a test for $x = 0$ before any division by $x$. 

\begin{definition}[Floating point numbers]
\label{def-float}
The set of floating point numbers with radix $\beta$ and $t$ 
digits of
mantissa is
\[
{\mathbf F}_{\beta, t} = \{ 0 \} \cup \{ m \beta^e: m,e \in \mathbb Z
\text{ and } \beta^{t} \le |m| < \beta^{t+1}\} 
\subset \mathbb R.
\]
\end{definition}
The set $F_{\beta, t}$ is not a \changed{subring}\label{ref16} of $\mathbb R$. The
union of all $F_{\beta, t}$ for $t \in \mathbb N$ is a \changed{subring}
of $\mathbb Q$ and is dense in $\mathbb R$.
The reader should check that the distance between two
consecutive non-zero floating point numbers satisfies
\begin{equation}\label{dist-next-float}
\beta^{-t-1} |x| \le |x-y| \le \beta^{-t} |x|
\end{equation}

For each $\epsilon > 0$ we abreviate
${\mathbf F}_{\epsilon} =
{\mathbf F}_{\beta, 1+\lfloor -\log_{\beta} (2\epsilon) \rfloor}$.
In particular, ${\mathbf F}_{2^{-t}}={\mathbf F}_{2, t}$.
With this definition, the distance between two consecutive 
non-zero floating
point numbers in ${\mathbf F}_{\epsilon}$ always satisfies
\[
2 \epsilon \beta^{-2} |x| < |x-y| \le 2 \epsilon |x|
\]

The computer arithmetic
model is specified through a 
family of rounding-off operators $(\fl_{\epsilon})_{\epsilon \in [0,1/4]}$,
\[
\function{\fl_{\epsilon}} 
{\mathbb R}{{\mathbf F}_{\epsilon}} 
{z}{\fl_{\epsilon}(z)}
\]
where $w=\fl_{\epsilon}(z)$ is such that $|z-w|$ is minimal. 
It may happen that some $z \in \mathbb R$ is of the form
$z=(x+y)/2$ for consecutive $x,y \in {\mathbf F}_{\epsilon}$.
In that case $w=x$ if the mantissa of $x$ is even, $w=y$ otherwise.
Also, we define $\fl_0(x)=x$.
In particular, $\fl_0(x)=x$ and ${\mathbf F}_0 = \mathbb R$. 
\medskip

By construction the rounding-off operators satisfy the following
properties:

\begin{eqnarray} 
\label{a-eps}
\forall x \in \mathbb R, && \fl_{\epsilon}(x) = x(1+\epsilon_1) 
\text{ with }|\epsilon_1| \le \epsilon.
\\
\label{a-mon}
\forall a, b \in \mathbb R,&& a \le b \Rightarrow \fl_{\epsilon}(a) \le \fl_{\epsilon}(b).
\\
\label{a-sym}
\forall a \in \mathbb R,&& \fl_{\epsilon}(-a) = -\fl_{\epsilon}(a).
\\
\label{a-rep}
\forall x \in \mathbb R,&&
\fl_{\epsilon}(\fl_{\epsilon}(x))=\fl_{\epsilon}(x)
\\
\label{a-ref}
\forall 0 \le \epsilon \le \delta < 1/4, \forall x \in \mathbb R,&&
x = \fl_{\delta}(x) \Rightarrow x = \fl_{\epsilon}(x).
\end{eqnarray}

Property \eqref{a-eps} is known as the {\em $1+\epsilon$-property}.
Properties \eqref{a-mon} and \eqref{a-sym} are known as {\em monotonicity}
and {\em symmetry}. Property \eqref{a-rep} states that there are
{\em representable} numbers, and property \eqref{a-ref} states that
representable numbers are still representable in higher precision.

Once the rounding-off operator is defined, it is assumed that elementary
operations $\circ \in \{+,-,\times,/\}$ are performed exactly and
then the result $z=x \circ y$ is replaced by $\fl_{\epsilon}(z)$ at
each step.  The exact result of some calculations is also a 
representable number. Moreover, it is possible to perform multiple precision
operations at a cost \changed{(see Remark \ref{distillation}).
\begin{lemma}\label{lemma-sum}
	\changed{Assume radix $\beta = 2$.}
Let $a, b \in {\mathbf F}_{\epsilon}$, $|a|\ge |b|$.
Let $c=\fl_{\epsilon}(a+b)$,
	$d=\fl_{\epsilon}(c-a)$, $e=\fl_{\epsilon}\changed{(b-d)}$. Then
$a+b = c+e$ exactly.
\end{lemma}
Before going further, we point out that this model for floating point
computations is stronger than
the one by \ocite{Cucker}. He assumed only the $1+\epsilon$
property \eqref{a-eps} and deterministic rounding off.
Both floating point models allow to test for positivity of a number $x$ or to
compare two given numbers. Our model has the advantage to allow
the comparison of  
$c$ with $a-b$ for exactly representable $a$, $b$ and $c$.
To avoid complications we assume that
$a$, $b$ and $c$ are strictly positive. Also, we can swap $b$
and $c$ in the expression above so we assume $b \le c$. 

\changed{\begin{lemma}\label{abc} Assume radix $\beta=2$.
Let $a, b, c \in {\mathbf F}_{\epsilon}$
be strictly positive, with $b \ge c$. Let $d=\fl_{\epsilon}(a-b)$ and
$e=\fl_{\epsilon}(d-c)$. Then,
\[
\sgn(e)
=
\sgn(a-b-c)
\]
\end{lemma}}
\changed{The procedure in Lemma~\ref{lemma-sum}
is also known as `compensated sum' or `Fast2Sum'.
\ocite{Muller} present a proof using actual computer
arithmetic with subnormal numbers. They also give extensive
references.}
\changedA{For the sake of completeness, a proof of those well-known Lemmas in our simplified model can be found in Appendix~\ref{sec-technical}}.

}
\section{Weak and strong computations}

Recall from Definition~\ref{BSS} that an {\em exact computation} for
a BSS machine $\mathbb R$ with input $x$ is a sequence 
$((\nu(t), s(t))_{t \in \mathbb N_0}$ in $\{1,\dots, N\} \times \mathbb R_{\infty}$ 
satisfying $\nu(0)=1$, $s(0)=I(x)$ and for all $t \ge 0$,
\begin{eqnarray*}
\nu(t+1) &=& \beta(\nu(t), s(t)) \hspace{1em}
\\
s(t+1) &=& g(\nu(t), s(t)) 
\end{eqnarray*}

We will define approximate computations without changing the definition
of a {\em machine}. 
A machine $M$ to test $x \in X$ will take an expression 
$(x, \epsilon)$ as input. 
In this paper there are two types of {\em approximate computations}.
Definitions are similar to those for exact computations.

\begin{definition}\label{strong} Let $0 \le \epsilon < 1/4$. 
	The {\em strong $\epsilon$-computation} for $M$ with input $(x, \epsilon)$ 
	is the sequence $((\nu(t), s(t))_{t \in \mathbb N_0}$ 
in $\{1,\dots, N\} \times \mathbb R_{\infty}$ 
satisfying $\nu(0)=1$, $\nu(t) = \beta(\nu(t-1), s(t-1))$ and
\begin{enumerate}[(a)]
\item \[
s_i(0) = \fl_{\epsilon}(I_i(x)) \ \forall i.
\]
\item 
\[
s_{0}(t+1) = \fl_{\epsilon}(g_0 (\nu(t), s(t)))  
\]
\end{enumerate}
\end{definition}
The computation {\em terminates} if $\nu(t) = N$ eventually, and the 
execution time $T=T(x)$ is the smallest of such $t$. 
The terminating computation is said to {\em accept} 
input $x$ if $s_0(T)>0$ and to {\em reject} $x$ otherwise. 
It is said to {\em correctly decide whether $x \in X$} if it
accepts $x$ for $x \in X$ and rejects $x$ otherwise.

\begin{definition} \label{weak} Let $0 \le \epsilon < 1/4$. 
A {\em weak $\epsilon$-computation} for $M$ with input $(x, \epsilon)$ 
is a sequence $((\nu(t), s(t))_{t \in \mathbb N_0}$ 
in $\{1,\dots, N\} \times \mathbb R_{\infty}$ 
satisfying $\nu(0)=1$, $\nu(t) = \beta(\nu(t-1), s(t-1))$ and
\begin{enumerate}[(a)]
	\item \changedA{At time} $t=0$,
\[
|s_i(0) - I_i(x)| \le \epsilon |I_i(x)| \ \forall i.
\]
\item If $\nu(t)$ is a computation node not associated to a copy operation,
\[
|s_{0}(t+1)-(g_0 (\nu(t), s(t))| \le \epsilon |g_0 (\nu(t), s(t))| 
\]
\item 
Otherwise,
\[
s_{i}(t+1)=(g_i (\nu(t), s(t))).   
\]
\end{enumerate}
\end{definition}
The same terminology applies to terminating, accepting, rejecting and
deciding approximate computations. \changed{Notice that a strong
$\epsilon$-computation is always a weak $\epsilon$-computation, but
the reciprocal does not hold in general.
We} will use practically the same definitions for weak and strong
computations of algebraic decision circuits.

\begin{definition}
Let $0 \le \epsilon < 1/4$.
A {\em strong $\epsilon$-computation} for the algebraic decision circuit $C$
with input $x_1,\dots, x_{n-1},x_n=\epsilon$ and constants $y$ is a sequence 
$w_1, \dots, w_{\tau}$ satisfying
$w_i=\fl_{\epsilon}(x_i)$ for $1 \le i \le n$, $w_i=\fl_{\epsilon}(y_{i-n})$ 
for $n<i\le m$
and then $w_i=\fl_{\epsilon}(w_{j(i)} \circ_i w_{k(i)})$ or 
$w_i=S(w_{j(i)},w_{k(i)},w_{l(i)})$. 
\end{definition}

\begin{definition}
Let $0 \le \epsilon < 1/4$.
A {\em weak $\epsilon$-computation} for the algebraic decision circuit $C$
with input $x_1,\dots, x_{n-1},x_n=\epsilon$ and constants $y$ is a sequence 
$w_1, \dots, w_{\tau}$ satisfying
$|w_i-x_i| \le \epsilon |x_i|$ for $1 \le i \le n$, 
$|w_i-y_{i-n}| \le \epsilon |y_{i-n}|$ 
for $n<i\le m$
	and then $|w_i - (w_{j(i)} \circ_i w_{k(i)})| \changedA{\le} \epsilon |w_{j(i)} \circ_i w_{k(i)} |$ or $w_i=S(w_{j(i)},w_{k(i)},w_{l(i)})$. 
\end{definition}

\begin{remark}
The definitions above for strong and weak computations for machines
and circuits 
imply that any input or constant may be
subject to rounding-off. In the case of machines, a constant needed
more than once may be `saved' so that the error occurs only once.
\end{remark}

Lemma~\ref{abc} shows that a BSS machine can test
whether three real numbers $a,b,c$ already in the memory
satisfy $c \ge a-b$, in such a way that exact or strong
computations will always produce the correct branching.
There is no way to compare $c$ and $a-b$ correctly under
a weak computation of precision $\epsilon$, unless one assumes that 
$|c - (a-b)| > \epsilon |a-b|$. The example below also illustrates
the relative power of strong and weak computations. This example
is a {\em routine} rather than a {\em machine} because we assume
that the real number $x$ is already stored in the memory. Otherwise
there is no way to distinguish $x$ from any number 
in $\fl_{\epsilon}^{-1}(x)$. But if $x$ is supposed to be a floating
point number, its binary expansion can be computed by the algorithm below.
\begin{example}\label{ex-weak-strong}
\begin{trivlist}
\item {\tt Routine} {\sc BitExpansion} $(x)$.
\item {\bf If} $x=0$ {\bf then} {\tt Return 0}
\item {\bf If} $x>0$ 
\subitem {\bf then} $s \leftarrow 1$ 
\subitem {\bf else} $s \leftarrow -1$.
\item $y \leftarrow x s$.
\item $e \leftarrow 0$
\item $p \leftarrow 1$
\item {\bf While} $y < 1$ {\bf do}
\subitem $e \leftarrow e-1$.
\subitem $y \leftarrow 2y$.
\item {\bf While} $y \ge 2$ {\bf do}
\subitem $e \leftarrow e+1$.
\subitem $y \leftarrow y/2$.
\item \changedA{$d=0$}
\item {\bf While} $y > 0$ {\bf do}
\subitem {\bf if} $y \ge 1$ 
\subsubitem {\bf then} $f_d \leftarrow 1$
\subsubitem {\bf else} $f_d \leftarrow 0$
\subitem $y \leftarrow 2(y - f_d)$.
\subitem $d \leftarrow d+1$
\item {\bf If} $x - s 2^e \sum_{i=0}^{d-1} f_i 2^{-i} \ne 0$ {\tt return} $-1$.
\item {\tt Return} $d, s, e, f_0, \dots, f_{d-1}$.
\end{trivlist}
\end{example}
We assume that the radix $\beta$ is $2$.
Any strong computation with precision \changed {$\delta \le \epsilon$}\label{ref18}
will return the
correct `bit expansion' of $x \in {\mathbf F}_{\epsilon}$. This is due
to the fact that all the numbers produced during intermediate computations
are representable in ${\mathbf F}_{\delta}$, hence in ${\mathbf F}_{\epsilon}$. 
The last {\bf if} test is redundant in the strong setting.
Now suppose that a weak computation with precision $\epsilon$ terminates.
\begin{lemma} \label{weak-bits}
Suppose that the expression $\tilde x = s 2^e \sum_{i=0}^{d-1} f_i 2^{-i}$ 
	as above with $s=\pm 1$ and $f_i \in \{0,1\}$ is
evaluated by Horner's rule under a weak $\epsilon$-computation, and has
approximate result $x$. Then $|x - \tilde x|< O(d\epsilon)|x|$. 
\end{lemma}
In other words, given $\delta$, there is $\epsilon \in \Omega (\delta/d)$
so that any weak accepting $\epsilon$-computation for the machine above is
guaranteed to produce a floating point number $\hat x \in {\mathbf F}_{\delta}$
in `bit representation' so that
$\hat x = x(1+\delta')$, $|\delta'|<\delta$.

\section{\P and \NP}

At this point we introduce our definition of polynomial time
problems for approximate computations. Recall that a problem
is a pair $(X, \mu)$ where $X$ is a set and $\mu: \mathbb R^{\infty}
\rightarrow \mathbb R$ is a {\em condition number}, $\mu(x) \ge 1$. 
We want to
decide whether $x \in X$ in time that is polynomial \changed{in}
$\Size{x} = \Length{x} (1+\log_2 \mu(x))$.

\begin{definition}[Deterministic polynomial time]\label{fpP}
The class \fpP of problems recognizable in polynomial time is the
set of all pairs 
$(X, \mu)$ so that  
there is a BSS machine $M$ over $\mathbb R$ with input $(x,\epsilon)$
\changed{and rational constants,} with the following properties:
\begin{enumerate}[(a)]
\item There are polynomials $p_{\Arith}$ and $p_{\Prec}$ such that whenever 
$x \in X$, $\Size{x} < \infty$ and $\epsilon < 2^{-p_{\Prec}(\Size{x})}$, 
		then the {\bf strong}
$\epsilon$-computation for $M$ with input $(x, \epsilon)$ terminates
in time $T< p_{\Arith}(\Size{x})$, 
and correctly decides whether $x \in X$. 
\item If there is a terminating {\bf weak} $\epsilon$-computation for $M$ 
with input $(x, \epsilon)$, $\epsilon < 1/4$ and output $z$,
then
\[
z > 0 \Longleftrightarrow x \in X
\]
\end{enumerate}
\end{definition}

\changed{
\begin{remark}\label{ref33}
An equivalent formulation would be to ask for a machine with 
floating point constants.
The more general definition would require all
the real constants of $M$ to be
efficiently computable in the classical
	sense: \changedA{for} every constant $x$, there is a BSS machine over
${\mathbb F}_2$ with input $k$
	that produces a floating point number with $x \in {\mathbf F}_{\epsilon}$,
$\epsilon = 2^{-k}$, in time polynomial in $k$. Restricting the constants
to $\mathbb Q$ simplifies the argument of Theorem~\ref{real-transfer} below
without changing the classes $\fpP$ and $\fpNP$.
\end{remark}
}
\changedA{
\begin{example}
Consider the set $X=\{x \in \mathbb R: x \le 1\}$ 
with condition $\mu(x) = 1/|x-1|$. 
	We claim that $(X, \mu) \in \fpP$. The machine $M$ from Definition~\ref{fpP} is given by the pseudo-code below
\begin{trivlist}
\item {\tt Input} $x, \epsilon$.
\item $a \leftarrow 1$
\item $b \leftarrow \max(x, 8 \epsilon)$
\item $c \leftarrow \min(x, 8 \epsilon)$
\item $d \leftarrow a-b$
\item $e \leftarrow d-c$
\item {\bf If} $c \le 0$ {\bf then} {\tt Output} $d$.
\item {\bf If} $e>0$ {\bf then} {\tt Output} $+1$.
\item {\tt Output} $-1$.
\end{trivlist}
Condition (a) is easy to check assuming radix $\beta=2$: 
suppose that $x \in X$ and $\epsilon < \frac{1}{11\mu(x)}$.
Assume first that $c > 0$. In particular $b>0$ and Lemma~\ref{abc}
says that 
\[
\sgn(e) = \sgn( 1 - \fl_{\epsilon}(8 \epsilon) - \fl_{\epsilon}(x) )
\]
using that $1$ is exactly representable. Also, $\fl_{\epsilon}(8\epsilon)	= 8 \fl_{\epsilon}(\epsilon)$ exactly so
\[
1 - \fl_{\epsilon}(8 \epsilon) - \fl_{\epsilon}(x)
\ge
	1 - (x + 8 \epsilon)(1+\epsilon)
\ge
	1-x - 8\epsilon - (x+8\epsilon) \epsilon 
\ge     \mu(x)^{-1} - 11 \epsilon > 0.
\]
Therefore the input is accepted.
If $c \le 0$, then either $\epsilon = 0$ (trivial case) or $x \le 0$
(another trivial case).	

Now we check condition (b). Values of $a,b,c,d,e$ will denote the
values of the states during the weak computations. Suppose that one
of those weak computations is accepted. There are two possibilities.
If $c \le 0$ and $d > 0$, 
then either $x \le 0$ or $0 = \epsilon < x$. So we assume
$b > c > 0$. From the fact that $e > 0$ we can infer that
$d>c$ hence $(a-b)(1+\epsilon) > c$. That is,
\[
	a - b - c + \epsilon (a-b) > 0
\]
The constant $1$ is approximated by
$a \le 1+\epsilon$. Also, reading constant $8$, input $\epsilon$ and 
computing the product under weak computation produces a quantity 
$c' \ge 8 \epsilon( 1-\epsilon)^3$. After reading $x$ we obtain a quantity
	$b' \ge x(1-\epsilon)$, and $b+c=b'+c'$. Therefore,
\[
b + c \ge x (1-\epsilon) + 8 \epsilon(1-\epsilon)^3
\]
and 
\[
		a-b \le 1 + \epsilon - \max(x (1-\epsilon), 8 \epsilon(1-\epsilon)^3)
\le 1 + \epsilon - x (1-\epsilon)
.
\]
Putting all together,
\[
0 < 1 + \epsilon - x(1-\epsilon) -8 \epsilon (1-\epsilon)^3
	+ \epsilon (1 + \epsilon - x (1-\epsilon))
.
\]
If furthermore $x \ge 1$,
\begin{eqnarray*}
	0 &<&  1 + \epsilon - (1-\epsilon) -8 \epsilon (1-\epsilon)^3
	+ \epsilon (1 + \epsilon - (1-\epsilon))
	\\
	&=& \epsilon (2- 8 (1-\epsilon)^3 + 2 \epsilon)
\end{eqnarray*}
which leads to a contradiction since $\epsilon < 1/4$ always.
\end{example}
}
\begin{example} The Cantor middle-thirds problem $(\mathbb R \setminus C, \mu_C)$ is in $\fpP$.
Recall that the {\em tent map}  was defined by
\[
T(x) = \left\{ 
\begin{array}{ll}
3 x & \text{for $x \le \frac{1}{2}$} \\
3-3 x & \text{otherwise.}
\end{array}
\right.
\]
and that the Cantor set is the set of points whose iterates by the tent map remain in the interval $[0,1]$. 
	Assume that $x \not \in C$. Let $x_i$ correspond to a weak $\epsilon$-computation, \changed{we claim that
\[
		|x_{i+1}-T(x_i)| <
		12 \epsilon
.
\]
Indeed, if $1/2 \le x_i \le 1$,
\[
	x_{i+1} = (3 - 3x (1+\epsilon_1) )(1+\epsilon_2),\hspace{2em} |\epsilon_j|<\epsilon.
\]
Thus
\[
	|x_{i+1}-T(x_i)| \le |-3x \epsilon_1 + 3 \epsilon_2 - 3x \epsilon_2 + 3x \epsilon_1 \epsilon_2|
< 12 \epsilon
\]
as claimed.\label{ref19}} 
	
\changed{In Example ~\ref{Cantor}, the condition number for the Cantor set
was defined as
\[
\mu_C(x) = \frac{1}{\min( d(x,C), 1)}.
\]
We also defined $C_0= [0,1]$ and inductively, $C_k=T^{-1}(C_{k-1})$. If
$x \in C_{k-1} \setminus C_k$, then $\mu_C(x) \ge 2 \times 3^k$.

	Now assume that  
	$\mu_C(x) < 2 \times 3^{k+1}$. Then
	$x \not \in C_{k-1}$ and in particular $T^k(x) \not \in [0,1]$. 
	It follows that $x \in C_{l-1} \setminus C_{l}$ for some $l \le k$
	so $T^l(x) > 1$.

	We take
	$\epsilon>0$ small enough. Successive iterates of
	$x=x_0$ satisfy $|x_i - T^i(x)| < 12 \epsilon (1+3+\cdots+ 3^{i-1})
< 6 \times 3^i \epsilon$. In particular,
	$|x_l - T^l(x)| < 6 \times 3^l \epsilon$. By construction
	$d(T^l(x),C)=T^l(x)-1 > 3^l / \mu_C(x)$.
A choice of $\epsilon < 1/(6 \mu_C(x))$ will ensure that 
\[
x_l - 1 > \frac{3^l}{\mu_C(x)} - 12 \times 3^l \epsilon > 0
\]
so any weak $\epsilon$-computation accepts $x_0=x$.}	
	The same argument is valid for the Koch snowflake and other fractals defined through a map with bounded Lipschitz constant.
\end{example}

\begin{remark}
In the definition of \Piter by \ocite{Cucker} 
the arithmetic cost
depends also on $\log \epsilon^{-1}$. Extra precision may be expensive
when not needed. The case $\epsilon=0$ is not allowed. Our point of
view is that a machine $\tilde M$ as in definition~\ref{fpP} where
we allow $p_{\Arith}=p_{\Arith(\Size{x},\epsilon)}$ can be
	simulated by a machine $M$ with input $(x, \delta)$\changed{, $\delta < \epsilon/2$}\label{ref20}, that will successively
simulate $\tilde M$ with input $(x,1/8)$, $(x,1/16)$, \dots, 
$(x, 2^{\lceil \log_2(\delta)\rceil})$, $(x,\delta)$. Then the
running time for this machine $M$ will still be polynomially bounded
in $\Size{x}$ for all $x \in X$, \changed{independently} of $\delta$. 
\end{remark}

\begin{remark}\label{distillation} Let $\epsilon < \delta < 1/4$. Assume the input $x$
is representable in precision $\delta$. It is still possible to
simulate a strong $\epsilon$-computation through a strong
$\delta$-computation,
at a cost in the arithmetic complexity (see Lemma~\ref{lemma-sum}). 
\changed{Extended precision floating point
arithmetic is explained in
\cites{Priest, Muller, Joldes}.}
Recall however that the
input may be non-representable with precision $\epsilon$. The
restriction $\epsilon <2^{-p_{\Prec}(\Size{x})}$
in Definition~\ref{fpP}(a) ensures
that strong computations start always from a correctly rounded input,
up to sufficient precision. 
\end{remark}

\begin{definition}[Non-deterministic polynomial time]\label{fpNP}
The class \fpNP of problems recognizable in non-deterministic
polynomial time is the set of all pairs 
$(X, \mu)$ so that  
there is a BSS machine $M$ over $\mathbb R$ with input $(x,y,\epsilon)$
	\changed{and rational constants}, with the following properties:
\begin{enumerate}[(a)]
\item There are polynomials $p_{\Arith}$ and $p_{\Prec}$ such that whenever 
$x \in X$, $\Size{x} < \infty$ 
there is $y \in \mathbb R_{\infty}$ so that
for all $\epsilon < 2^{-p_{\Prec}(\Size{x})}$, 
the {\bf strong}
$\epsilon$-computation for $M$ with input $(x,y,\epsilon)$ terminates
in time $T< p_{\Arith}(\Size{x})$ and correctly decides whether $x \in X$. 
\item If there is some $y \in \mathbb R_{\infty}$ for which there is
a terminating {\bf weak} accepting
$\epsilon$-computation for $M$ 
with input $(x, y, \epsilon)$, $\epsilon < 1/4$, then  
$x \in X$.
\end{enumerate}
\end{definition}


\begin{remark} If one replaces the word
{\bf strong} by {\bf weak} in Definitions~\ref{fpP} and~\ref{fpNP},
then one respectively obtains
the classes \Piter and $\NPiterU$ as defined by
\cite{Cucker}.
\end{remark}

\begin{remark}
	\changedA{One reason for} using {\em weak} computations in 
Definitions~\ref{fpP}(b) and~\ref{fpNP}(b) is to allow for natural \NP
problems that may be defined by \changedA{algebraic equalities
and inequalities}. 
More precisely, if $(X,\mu) \in \fpNP$, then $X \cap \{x : \mu(x)< \infty\}$
is a countable union of semi-algebraic sets.
Eventually
the condition number may look less natural as for $\CircPseudoFeas$ next
	section. It is an open question whether replacing {\em weak} by
	{\em strong} in those definitions would actually enlarge the
	complexity classes $\fpP$ and $\fpNP$.
\end{remark}

\begin{example}[Feasibility for non-linear programming] {\em Given $(f_1, \dots, f_l)$ a $l$-tuple of $n$-variate real polynomials of degree $\le D$, decide if there is a point $x \in \mathbb R^n$ with $f_1(x), \dots, f_l(x) > 0$.} One should specify the input format: each polynomial is a finite list $(f_{ia}, a)_{a \in A_i}$, $f_{ia} \in \mathbb R$ and $A_i$ a finite set 
	of $n$-tuples of non-negative integers. It is assumed that each $A_i$ contains the origin. 
		We will use the $1$-norm on the space of polynomial systems with the same structure: 
\[
	\|f-g\|_{1} = 
		\max_{i} \sum_{a\in A_i} |f_{ia}-g_{ia}|.
\]
		
		In order to specify the input size we must attribute a condition number $\mu(f)$ to each system. If there are no values of $x$ satisfying the simultaneous inequations, then $\mu(f)=\infty$. Otherwise,
\[
	\mu(f) = \frac{1}{\max\{\delta: \exists x \in \mathbb R^n, \|f-g\|_{1} \max(1,\|x\|_{\infty})^D <\delta \Rightarrow g_1(x), \dots, g_l(x) > 0\}}
\]
We claim that this problem is in $\fpNP$. The key observation is that a weak $\epsilon$-computation for $f(x)$ actually produces exactly some $g(x)$, where
\[
	|f_{ia}-g_{ia}| \le |f_{ia}| ((1+\epsilon)^{D+\#A_i}-1)
\]
The total forward error in the computation of $f_i(x)$ 
is therefore bounded above by
\[
	\|f\|_1 \max(1,\|x\|_{\infty})^D ((1+\epsilon)^{D+\#A_i}-1)
\]
Therefore, the machine has just to guess $x$ and guarantee that $\epsilon$ is small enough, namely 
\[
	g_i(x) > \|f\|_1  \max(1,\|x\|_{\infty})^D  ((1+\epsilon)^{D+\#A_i}-1)
\]
All this can be done with weak $\epsilon$ computations, with $\epsilon$ small enough. In order to guarantee that the tests succeed for a near-optimal $x$, one can set 
\[
\|f\|_1  \max(1,\|x\|_{\infty})^D ((1+\epsilon)^{D+\#A_i}-1) < \delta/3
\]
with $\delta = 1/(2\mu(f))$.
Since $f_i(x)-\delta > 0$, $g_i(x)> 2\delta/3$ and the test will succeed.
\end{example}

\changed{
	\begin{remark}\label{ref21}
	The example above uses only weak computations. Example \ref{pseudo-accepting} below will use weak
and strong computations.
\end{remark}
}

\section{$\NP$-completeness}
Black box machines appeared in Definition~\ref{oracle1} as a
tool for defining reductions.
In the context of numerical computations we assume that
a black box node for a problem $(Y,\eta)$
receives input of the form $(S,y,\epsilon)$ where $\epsilon$ is
the current precision. 
\changed{We can define
strong and weak computations almost as before.
If $\Size{y}<S$ and $\epsilon < 2^{-\Size{y}}$ 
then}
the strong $\epsilon$-computation 
will always return the correct answer whether $y \in Y$. A weak
computation will never return a wrong answer. The running time of
the composite machine is exactly as in Definition~\ref{oracle2}.
\begin{definition} \label{oracle3}
	The strong computations for a machine $M$ with a black box
for the problem
$(Y,\eta)$ 
	are the \changedA{same as} in Definition \ref{strong}, 
except when the black box 
node $\nu_O$ is attained.
If $\nu_O$ is first attained at $t=t_0$ and the machine is at internal state
	$s(t_0) = \changed{(z,S. y)}$, then
\[
(\nu(\tau), s(\tau)) = 
\left\{
\begin{array}{ll}
(\nu_0, \changed{(z,r. y)}) & \text{for $0 \le \tau - t_0 < S$;} \\
(\beta(\nu_0),\changed{(z,r. \dots)}) & \text{if $\tau-t_0=S$ with 
}
\\
\multicolumn{2}{c}{
	\changedA{r=\left\{\begin{array}{ll}
	+1 & \text{if $y \in Y$ and $\Size{y} \le S$,}\\
	-1 & \text{if $y \not \in Y$,}\\
	\pm 1 & \text{in all the other cases.} \end{array}\right.
	}}
\end{array}
\right.
\]
\end{definition}

\changed{
\begin{definition} \label{oracle4}
	The weak computations for a machine $M$ with a black box
for the problem
$(Y,\eta)$ 
	are the same as the same as in Definition \ref{weak}, except when the black box 
node $\nu_O$ is attained.
If $\nu_O$ is first attained at $t=t_0$ and the machine is at internal state
	$s(t_0) = \changed{(z,S. y)}$, then
\[
(\nu(\tau), s(\tau)) = 
\left\{
\begin{array}{ll}
	(\nu_0, \changed{(z,r. y)}) 
&
\text{for $0 \le \tau - t_0 < S$;}
\\
	(\beta(\nu_0),\changed{(z,r. \dots)}) & \text{if $\tau-t_0=S$ with
}
\\
\multicolumn{2}{c}{
	\changedA{
   r=\left\{\begin{array}{ll}
	\pm 1 & \text{if $y \in Y$,}\\
	-1 & \text{if $y \not \in Y$.}
   \end{array}\right.}}
\end{array}
\right.
\]
\end{definition}
}

\begin{definition}\label{fpTuring-reduction}
A polynomial time 
Turing reduction from $(X,\mu)$ to $(Y,\eta)$ is a BSS machine $R$
over $\mathbb R$ with \changed{rational constants} and a black box node for $(Y,\eta)$ such
that:
\begin{enumerate}[(a)]
\item There are polynomials $p_{\Arith}$ and $p_{\Prec}$ such that 
whenever $x \in X$, 
$\Size{x} < \infty$ and $\epsilon<2^{-p_{\Prec}(\Size{x})}$, 
		the {\bf strong} $\epsilon$-computation
for $R$ with input $(x, \epsilon)$ terminates in time
$T< p_{\Arith}(\Size{x})$ and accepts $x$. 
\item If there is accepting {\bf weak} $\epsilon$-computation for $R$ with 
input $(x,\epsilon)$
$\epsilon < 1/4$,
then $x \in X$.
\end{enumerate}
\end{definition}

\begin{definition}
A problem $(Y, \eta)$ is $\fpNP$-hard if and only if
for any problem $(X, \mu) \in \fpNP$, there is a polynomial time 
Turing reduction
from $(X, \mu)$ to $(Y, \eta)$. A problem $(Y, \eta)$ is $\fpNP$-complete
if $(Y, \eta) \in \fpNP$ and $(Y, \eta)$ is $\fpNP$-hard.
\end{definition}

\begin{theorem}\label{fp-complete}
If there is a $\fpNP$-complete problem in $\fpP$, then $\fpP=\fpNP$.
\end{theorem}

\changed{
\begin{proof}

Assume that $(Y, \eta)$ is $\fpNP$-complete, and that it is in $\fpP$.
	Let $M_Y$ be the machine solving $Y$ \changedA{in polynomial time}. We have to prove that for any
problem $(X, \mu) \in \fpNP$, $(X, \mu) \in \fpP$. 

Because $(Y, \eta)$ is $\fpNP$-complete, there is a Turing reduction
from $(X, \mu)$ to $(Y, \eta)$. Let $M$ be the machine with black box
node $\nu_O$ for the Turing reduction.

	We produce a machine $M_X=M_X(x,\epsilon)$ by replacing each call to $\nu_O$ at state
	$(z, S.y, \epsilon)$ 
by a call to the machine
$U'$ from Proposition~\ref{timed} with input $(S, f_{M_Y}, y, \epsilon)$.
Properly speaking, one should modify $U'$ in order to make sure that the 
contents of $z$ are preserved.

There is a polynomial $p_1$ with the following property: if $x \in X$
and $\epsilon < 2^{-p_1(\Size{x})}$ then any strong 
$\epsilon$-computation of the machine $M$ succeeds, in time 
	at most	$\changedA{q(\Size{x})}$ for some polynomial $q$. 

We need to show that for $\epsilon$ small enough, any strong computation
of $M_X$ with input $(x, \epsilon)$ will succeed. Assume for the time being
	that $\epsilon < 2^{-p_1\changedA{(\Size{x})}}$. Then the strong $\epsilon$-computation
of $M_X$ with input $(x, \epsilon)$ restricts to a strong 
$\epsilon$-computation of $M$ with input $(x, \epsilon)$ 
and $t \ge 0$ calls to the black box routine.

Let $S_1, \dots, S_t$ be the size bounds and $S = \max S_i$.
	The polynomial time bound $q(\Size{x})$ for the black box machine $M$
with input $(x, \epsilon)$, is also a polynomial time bound for
$\sum S_i$, hence for both $S$ and $t$. 
In $M_X$, each call of the black box is replaced by a call to the
machine $U'$ with input some input $S_i, y, \epsilon$. 

The machine $U'$ is actually simulating the machine $M_Y$ with
input $(y, \epsilon)$ up to $S_i$ steps.
Because $(Y, \eta) \in \fpP$, there is a polynomial $p_2$ so that
if $y \in Y$ and $\epsilon < 2^{-p_2(\Size{y})}$ then any
strong $\epsilon$-computation of $M_Y$ will succeed.

In particular, if $y \in Y$, $\Size{y} \le S_i$ and $\epsilon < 2^{-p_2(S)}$, 
then the simulation will succeed. We can now set $p_{\Arith}(t) =
p_1(t)+p_2(q(t))$ and assume that $\epsilon < 2^{-p_{\Arith}(\Size{x})}$.
A strong $\epsilon$-computation for the machine $U'$ will accept the 
input $(S,y, \epsilon)$ whenever $y \in Y$ and
$\Size{y} \le S$. Because strong computations are also weak computations, 
a strong $\epsilon$-computation for the machine $U'$ with input $(S,
y, \epsilon)$ and $y \not \in Y$ will always reject the input. The other
cases are irrelevant: the machine $U'$ is behaving as the black box in
the register equations of $M$, with a larger running time.

The running time for each call to $U'$ is polynomially
bounded in $S$, and hence in $\Size{x}$. It follows that if $x \in X$ then
the
running time of $M_X$ with input $(x,\epsilon)$ is bounded above
by some polynomial $p_{\Arith}(\Size{x})$.

Finally, we need to check that no weak $\epsilon$-computation of $M_X$
can accept some input $(x, \epsilon)$ with $x \not \in X$. Because no
$\epsilon$-computation of $M$ can do that, we need to assume that there
is an $\epsilon$-computation for $U'$ in one of the queries that wrongly
returns $y \in Y$ (see definition~\ref{oracle4}). Therefore, there is
a weak $\epsilon$-computation for $M_Y$ that accepts some $y \not \in Y$,
contradiction.

Thus, the machine $M_X$ satisfies definition~\ref{fpP} and 
the problem $(X,\mu)$ is in $\fpP$ as claimed.
\end{proof}
}

\begin{example}\label{pseudo-accepting}
The pseudo-accepting set of an algebraic decision circuit $C$ is
\[
\begin{split}
	\mathcal A(C) = \left\{ (x,w):\right. &\text{ $w$ is an accepting weak $\delta$-computation} \\
& \left. \text{of $C$ with input $(x, \delta)$, for some $\delta < 1/4$}\right\}
.\end{split}
\]
We define $\CircPseudoFeas$ (Circuit Pseudo-Feasibility) as the set of 
	algebraic circuits $C$ with \changedA{rational constants that have pseudo-accepting set $\mathcal A(C) \ne \emptyset$}. 
	To show that a given $C \in \mathcal A(C)$ one may want to exhibit a witness $(x,w,\delta)$ so that $(x,w)$ is an accepting weak $\delta$-computation. Establishing that an input is an \changed{actual}\label{ref22} weak $\delta$-computation using approximate computations is a thorny matter. We will settle for a test that passes when $(x,w)$ is a $\delta/2$ weak computation, $\delta < 1/8$. Success of the test will imply that $(x,w) \in \mathcal A(C)$. Not all elements of $\mathcal A(C)$ will admit such a witness, and those not admitting such a witness will be deemed to be unstable and ill-posed.
The condition number for $\CircPseudoFeas$ will measure the amount of information necessary to check a witness with a machine with precision $\epsilon$. The smaller is $\epsilon$, the larger will be the condition. Unstable witnesses will have infinite condition. Formally, we define the condition
	number $\mu(C)=\rho(C)^{-1}$, where
	\begin{equation}\label{def-rho}
\begin{split}
	\rho(C) = \sup \{& \epsilon>0:  
\exists (x,w) \in \mathcal A(C),
	w \in {\mathbf F}_{\epsilon}^{\infty},\\& \ \exists \ \delta \changed{\text{ with }} \epsilon<\delta<1/8,
	\text{and $w$ is an accepting}\\& \text{weak $\delta/2$-computation 
of $C$ with input $(x, \delta)$} \}.
\end{split}
\end{equation}
We make the convention that $\rho(C) = 0$ if the set above is empty. Notice that
$\rho(C) > 0$ implies that $C \in \CircPseudoFeas$.
\end{example}
This definition makes ill-posed the circuits that admit an exact
accepting computation but no accepting non-exact weak 
$\delta$-computations for
$\delta > 0$. 
There is no reason for some $C \in \CircPseudoFeas$ to actually accept
any exact computation. Moreover, $\rho(C) \le 1/8$ by construction.

\begin{theorem}\label{th-main} $(\CircPseudoFeas,\rho^{-1})$ is $\fpNP$-complete.
\end{theorem}

\changedA{If a circuit has rational constants, its length is assumed to be the number of nodes plus the bit length of the constants. The proof of Theorem~\ref{th-main} will use the} proposition below:

\begin{proposition}\label{prop-technical}\label{ref25}
There is a machine $U$ with input $(C,x,w,\delta,\epsilon)$
such that:
\begin{enumerate}[(a)]
\item If the following conditions hold:
\begin{enumerate}[i.]
\item  $0 \le \delta < 1/8$, 
\item  $0 \le \epsilon < \delta/32$,
\item  $C$ is \changed{an
	algebraic circuit with \changedA{rational} constants}, and 
\item  $w$ is an accepting weak
$\delta/2$-computation for
$C$ with input $(x, \delta)$, 
with each $w_i \in {\mathbf F}_{\epsilon}$.
\end{enumerate}
		then the strong $\epsilon$-computation for
$U$ with input $(C,x,w,\delta,\epsilon)$ accepts the 
input in $\Length{C}^{O(1)}$ steps.
\item If for some $\epsilon < 1/4$ 
there is a weak $\epsilon$-computation for $U$ accepting the input
$(C,x,w,\delta,\epsilon)$, then there is a weak 
		accepting \changed{$\delta < 1/7$}-computation for $C$ with input $(x, \delta)$.
\end{enumerate}
\end{proposition}
The reason to require $w \in {\mathbf F}_{\epsilon}$ at input in (a) is to address the problem of deciding queries such as $w_1 \ge w_2$ or $w_1 - w_2 > w_3$ for inputs $w_i$ known up to precision $\epsilon$. This is impossible using weak computations, but it is possible using strong computations with the hypothesis that the inputs are in $\mathcal F_{\epsilon}$. See Lemmas~\ref{lemma-sum} and \ref{abc}.
The proof of Proposition~\ref{prop-technical} is postponed to the next section. 
We prove Theorem~\ref{th-main} in two steps.

\begin{proof}[Proof that the problem $(\CircPseudoFeas,\rho^{-1})$ is in $\fpNP$]
We claim that the machine $U=U(C,y,\epsilon)$, $y = (x,w,\delta)$ 
from Proposition~\ref{prop-technical}
has the properties required by Definition~\ref{fpNP}. 
\par
\noindent
{\bf Property (a):} {\em 
There are polynomials $p_{\Arith}$ and $p_{\Prec}$ such that whenever 
$C \in \CircPseudoFeas$, $\Size{C} < \infty$ 
there is $y=(x,w,\delta) \in \mathbb R_{\infty}$ so that
for all $\epsilon < 2^{-p_{\Prec}(\Size{C})}$, 
the {\bf strong}
$\epsilon$-computation for $U$ with input $(x,y,\epsilon)$ terminates
in time $T< p_{\Arith}(\Size{x})$ and correctly decides whether $x \in X$.}

Let $p_{\Prec}(S) = S+1$. By definition of $\rho(C)$,
there are some $\rho(C)/2 < \epsilon' \le \rho(C)$ and 
$(x,w) \in \mathcal A(C)$ so that $w$ is an accepting weak 
$\delta/2$-computation of $C$ with input $(x,\delta)$ for  
$\epsilon' < \delta < 1/8$, and $w \in {\mathbf F}_{\epsilon'}$.
We pick $\epsilon$ small enough, 
\[
\epsilon < 2^{-p_{\Prec}(\Size{C})} 
=
2^{- \Length{C} (1 - \log_2( \rho(C)) ) -1}
.
\]
The bounds $\Length{C} \ge 2$, 
	\changedA{$\rho(C) < 2 \epsilon'$} and $\rho(C) \le 1/8$ imply that
\[
\epsilon < 2^{-3 +\log_2(\rho(C)^2)} 
	\changedA{=\frac{\rho(C)^2}{8} < \frac{ \epsilon'}{32}	
	< \frac{\delta}{32}.}
\]
Since $\epsilon < \epsilon'$,
${\mathbf F}_{\epsilon} \supseteq {\mathbf F}_{\epsilon'}$ and we are
in the hypotheses of Proposition~\ref{prop-technical}(a). It follows
that the strong $\epsilon$-computation of $U$ with input $(C,y,\epsilon)$,
$y=(x,w,\delta)$ terminates in time $\Length{C}^{O(1)} < \Size{C}^{O(1)}$.

{\bf Property (b):} {\em 
If there is some $y \in \mathbb R_{\infty}$ for which there is
a terminating {\bf weak} accepting
$\epsilon$-computation for $U$ 
with input $(C, y, \epsilon)$, $\epsilon < 1/4$, then  
$C \in \CircPseudoFeas$.}

This follows directly from Proposition ~\ref{prop-technical}(b).
\end{proof}

\begin{proof}[Proof that $(\CircPseudoFeas,\rho^{-1})$ is \fpNP-hard]
Let $(X, \mu) \in \fpNP$ be recognized by the non-deterministic
polynomial time machine $M=M(x,y,\delta)$ of  Definition~\ref{fpNP}.
Let $p_{\Arith}$, $p_{\Prec}$ be the polynomial bounds for the definition.
We will now produce a reduction from $(X,\mu)$ to the problem
$(\CircPseudoFeas,\rho^{-1})$.

For every $T>0$, $C_{M,T,x}=C_{M,T,x}(y,\delta)$ is the 
time-$T$ directed algebraic circuit
for the time-$T$ computation of $M$ with input $(x,y,\delta)$. The
$x_i$'s and the constants of $M$ 
are treated as constants of the circuit $C_{M,T,x}$.

Unlike the proof of Theorem~\ref{CircFeasNPc}, we need to make a distinction
between discrete and continuous variables when building the circuit.
Recall that in this paper discrete variables are represented by a real 
number $x$, where $x>0$ is interpreted as `Yes' or $1$ and $x \le 0$
is interpreted as `No' or $0$. Discrete or boolean operations during exact, 
strong or weak $\epsilon$-computations are always correct. This will allow to
reduce time-$T$ computations to an algebraic decision circuit where the
discrete computations are correct and the floating point computations are
exactly the floating point computations from the machine $M$ with input
$x,y,\delta$.

We give two constructions for this circuit. One can keep the original
Cucker-Torrecillas, as in Theorem~\ref{CircFeasNPc} above. 
Let $(\nu(t), s(t))_{t \in \mathbb N_0} \in \{1,2, \dots, N \} 
\times \mathbb R_{\infty}$
denote the {\em node} and {\em state} at time $t$.
The node $\nu(t)$ will be represented by $d=\lceil \log_2 N \rceil$
`discrete' variables in $\{0,1\}$, that is $d$-bits.

All the Lagrange interpolating polynomials 
\[
L_{i}(\nu)= \prod_{j\ne i} \frac{\nu-i}{j-i} 
\]
must be computed with discrete variables. In order to compute
\[
	s_k(t) = \sum_{i=1}^N L_i(\nu(t-1)) \changed{g_k}(\nu(t-1), s(t-1))
\]
one should use the fact that $L_i(\nu(t-1)) \in \{0,1\}$ and insert
	a selector with input $T_i=\changed{g_k}(\nu(t-1), s(t-1)), T_{i-1}, L_i(\nu(t-1))$,\label{ref26}
assuming also $T_{-1}=0$. Then,
\[
s_k(t) = T_N
\]
without introducing extra floating point operations.
The second reduction using selectors and no Lagrange interpolation is
given in Appendix~\ref{alt-reduction}.

Both reductions provide  
a circuit $C_{M,T,x}$ with length bounded by a polynomial $r(T)=r_M(T)$.
Given $M$, $x$, there is a routine that produces a representation of
the circuit $C = C_{M,T,x}$ in time polynomial in $T$. No floating
point operation appears in that routine, except possibly for a copy
operation \changed{which}\label{ref27} is always exact.

Let $\nu_O$ be the black box associated to the problem 
$(\CircPseudoFeas, \rho^{-1})$.
The Turing reduction from $(X,\mu)$ to $(\CircPseudoFeas, \rho^{-1})$
is the BSS machine $R$ with 
black box $\nu_O$ given by the 
\changed{pseudo-code below. Recall that if the black box $\nu_O$ is 
attained when the machine is in some state $(z,S.y)$, then it takes
$S$ as a tentative upper bound for $\Size{y}$. If $y$ is not too large
and $y$ belongs to $\CircPseudoFeas$, then it eventually switches to a state of
the form $(z,1.\dots)$. 
We say in this case that `$\nu_O$ accepts $(S,y)$'.
The state variables in $z$ are preserved, so the machine can save whatever
is important for further computations.
The black box machine may also switch to state $(z,-1.\dots)$ if
$y \not \in Y$ or if \changedA{$\Size{y}$} is too large, in that case the input is rejected.
}
\begin{trivlist}
\item {\tt Input} $(x,\delta)$.
\item $T \leftarrow \Length{x}$.
\item {\bf Repeat}
\subitem $T \leftarrow 2T$
\subitem {\tt Produce $C=C_{M,T,x}$.}
	\subitem {\bf If }{\tt the black box $\nu_O$} accepts $(\changed{1}+(T+2)r(T),\changed{(C,\delta)})$\label{ref28}
{\bf then} {\tt output} 1.
\end{trivlist}

\changed{Later on, we will establish that $\changed{1+}(T+2)r(T)$ will eventually be
an upper bound for the size of $(C,\delta)$ for $T$ large enough.}
Recall from \changed{Definition~\ref{oracle3}}\label{ref29} that 
\changed{if $\Size{C} < (T+2)r(T)$ and $\delta < 2^{-\Size{C}}$ then the strong $\delta$-computation answers in time bounded by $\Size{C}$}
whether the circuit $C \in \CircPseudoFeas$, that is whether $C$ 
has an accepting weak $\delta$-computation with input $(y,\delta)$,
$\delta < 1/4$.
Reciprocally, \changed{the pseudo-code above admits no accepting 
weak $\delta$-computation when}
$C \not \in \CircPseudoFeas$.

We check now that $R$ is a polynomial time 
Turing reduction (Def.\ref{fpTuring-reduction}).
Let $x \in X$, $\Size{x}<\infty$. 
Assume that $\delta < \delta_0 = 2^{-T_1}$ with
\[
T_1=
p_{\Prec}
\changed{
	\left(\Size{x} (1+r(p_{\Arith}(\Size{x})))\right)+4.
}
\]
In part. $\delta_0 \le 1/16 < 1/8$.
Definition~\ref{fpTuring-reduction}(a) requires that the strong $\delta$-computation 
for the
black box machine $R$ with input $(x,\delta)$ accepts the input,
and terminates in time polynomial \changed{in} $\Size{x}$.

Because $\delta < 2^{-p_{\Prec}(\Size{x})}$, there is $y$ so that
the strong $\delta$-computation
for $M$ with input $(x,y,\delta)$ accepts the input in time
\[
T_0 \le p_{\Arith}(\Size{x})
.
\]
Eventually the variable $T$ from machine $R$ will be larger than
$T_0$.
At this point the strong $\delta$-computation for 
the circuit $C=C_{M,T,x}$ with input $(y,\delta)$ will be accepting. 
Therefore $C \in \CircPseudoFeas$.

To actually show that $x$ is accepted in polynomial time, we assume
that $\changedA{T_0 \le}$ $T_1 \le T < 2T_1$. We bound $\Size{C}$ as follows:
let $\epsilon = \delta_0/2$ $\changedA{=2^{-T_1-1}}$. Then a strong $\epsilon$-computation for
$C$ with input $(y,\delta_0)$ is also accepting, and we have by construction
$\epsilon < \delta_0 < 1/8$. Therefore $\rho(C) > \epsilon$, hence
\[
\log_2 \rho(C)^{-1} < T_1+1 \le T+1.
\]
Recall that the circuit $C=C_{M,x,T}$ does not depend on $\delta$.
Therefore for any input $(x,\delta)$ to $R$, once $T_1 \le T < 2T_1$,
\[
\Size{C} = 
(1+ \log_2\rho(C)^{-1}) 
\Length{C} 
\le
(2+T) r(T)
\]
and the input $(\changed{1+}(2+T)r(T),\changed{(C,\epsilon)})$ gets accepted by the black box node.
Since $T \le 2T_1$ and $T_1$ is polynomially bounded by $\Size{x}$, any strong
$\delta$-computation with $\delta < \delta_0$ will accept $x$ within time
polynomial \changed{in} $\Size{x}$.

It remains to prove that no weak $\delta$-computation for $R$ with input
$(x, \delta)$, $\delta < 1/4$, will accept $x \not \in X$. Indeed,
the only accepting possibility is for the black box to accept an input
$(\changed{1+}(2+T)r(T),\changed{(C,\delta)})$ for $C=C_{M,x,T}$ and this can only happen
for $C \in \CircPseudoFeas$. There should therefore exist an accepting
weak time-$T$ computation for $M$ with input $x,y,\delta$, $\delta < 1/4$.
This can only happen if $x \in X$.
\end{proof}

\section{Proof Proposition~\ref{prop-technical}}

The objective of this section is to prove Proposition~\ref{prop-technical}.
For the sake of concision, the command 
\begin{trivlist}
\item 
\hspace{2em}
{\bf check }$\langle${\em expression}$\rangle$ 
\end{trivlist}
means
\begin{trivlist}
\item 
\hspace{2em}
{\bf If } $\neg$ $\langle${\em expression}$\rangle$ {\bf then }{\tt output -1}.
\end{trivlist}
Recall that this is the pseudo-code for a BSS-machine. Strong and weak
$\epsilon$-computations may approximate or not the quantities in the
expression. The pseudo-code for $U$ is:

\begin{enumerate}[1]
\item \hspace{2em} {\tt Input} $(C,x,w,\delta,\epsilon)$.
\item \hspace{2em} {\bf Check} $\delta \le \frac{1}{8}$.
\item \hspace{2em} {\bf Check} $\epsilon \le \frac{1}{32}\delta$.
\item \hspace{2em} $C_1 \leftarrow 1+\frac{3}{4} \delta$.
\item \hspace{2em} $C_2 \leftarrow 1-\frac{3}{4} \delta$.
\item \hspace{2em} {\bf For each} $i \in \{1, \dots, \Length{w}\}$,
\item \hspace{4em} {\bf If} vertex $i$ is input or constant with value $c$
\item \hspace{6em} {\bf then if }$c \ge 0$ {\bf then check} $C_2 w_i \le c \le C_1 w_i$.
\item \hspace{6em} {\bf else check} $C_1 w_i \le c \le C_2 w_i$.
\item \hspace{4em} {\bf If} vertex $i$ is an arithmetic operation $\circ_i$ and incident edges from $j$ and $k$
\item \hspace{6em} {\bf then if }$w_i \ge 0$ {\bf then check} 
$C_2 w_i \le w_j \circ_i w_k \le C_1 w_i$
\item \hspace{6em} {\bf else check} 
$C_1 w_i \le w_j \circ_i w_k \le C_2 w_i$
\item \hspace{4em} {\bf If} vertex $i$ is a selector with incident edges from $j$, $k$ and $l$
\item \hspace{6em} {\bf then check} $w_i = S(w_j,w_k,w_l)$.
\item \hspace{2em} {\tt Output 1}.
\end{enumerate}

\begin{lemma}\label{lemma-delta}
Assume that $\epsilon < 1/4$.
\begin{enumerate} 
\item If $\delta < 1/8$,
	then the strong $\epsilon$-computation for $U$ passes the test
at line 2.
\item Assume that a weak $\epsilon$-computation for $U$ passes
the test of line 2. Then $\delta \le 5/24$.
\end{enumerate}
\end{lemma}

\begin{proof}\label{ref31}
\changed{Because of monotonicity \eqref{a-mon}, if $\delta < 1/8$ then 
$\fl_{\epsilon}(\delta) \le \fl_{\epsilon}(1/8)$,
	so the strong $\epsilon$-computation\label{ref32} for $U$ passes the test in line 2.
}
	\changed{
Now, assume that there is a weak $\epsilon$-computation for $U$
that passes the test of line 2. Reading input and constants
introduces error, so 
\changedA{we are actually comparing $\delta(1+\epsilon_1)$
with $\frac{1}{8} (1 + \epsilon_2)$, $|\epsilon_i|< \epsilon$.
	We can conclude that}
\[
	\delta (1-\epsilon) < \frac{1+\epsilon}{8}
\]
so $\delta \le \frac{5}{24}$.}
\end{proof}

\begin{lemma}\label{lemma-epsilon}\label{ref30}
\begin{enumerate} 
\item If $\delta < 1/8$ and $\epsilon < \delta/32$,
	then the strong $\epsilon$-computation for $U$ passes the test
at line 3.
\item 
Assume that $\epsilon < 1/4$ and 
that a weak $\epsilon$-computation for $U$ passes
the tests of lines 2 and 3. Then $\epsilon < 1/250$, $\delta<1/7$ and
		$\epsilon < \changed{\frac{\delta}{31}}$.
\end{enumerate}
\end{lemma}

\begin{proof}
	\changed{As in Lemma~\ref{lemma-delta},
if $\delta < 1/8$ then 
	the strong $\epsilon$-computation for $U$ passes the test in line 2.}
Under the hypotheses of item (1), $\epsilon < 2^{-8}$.
Assuming radix $2$ computations, $32=2^8$ is exactly representable and
dividing by $32$ is an exact operation. By monotonicity,
\[
\fl(\epsilon) \le 
\fl( \fl(\delta)/32 ) 
=
\fl(\delta)/32 
\]
and the test in line 3 passes.
\medskip
\par

Now, assume that a weak $\epsilon$\changed{-computation} 
passes the tests
at lines 2 and 3. \changed{Reading $\epsilon$ produces
a number in the open interval
$(\epsilon (1-\epsilon), \epsilon (1+\epsilon))$. 
Reading $\delta$, the constant $1/32$ and multiplying together
produces a number in the interval
$( (1-\epsilon)^3 \delta/32, (1+\epsilon)^3 \delta/32)$.
Because the test in line 3 passed,
}
\[
\epsilon (1-\epsilon) \le \frac{\delta}{32}(1+\epsilon)^3
\]
and hence
\[
\epsilon \le \frac{\delta}{32}\frac{(1+\epsilon)^3}{1-\epsilon}
\]
for $\delta < \frac{1}{8} \frac{1+\epsilon}{1-\epsilon}$.
We deduce that $\epsilon \le f(\epsilon)$ where
\[
f(\epsilon) = \frac{1}{256} \frac{(1+\epsilon)^4}{(1-\epsilon)^2}
.
\]
We know that $\epsilon < \epsilon_0 = 1/4$. Iterating $\epsilon_{i+1} \le f(\epsilon_i)$, we obtain $\epsilon < \epsilon_3 \simeq 0.003970515\cdots < 1/250$. 

If a weak computation passes line 3, then  
\[
\delta (1-\epsilon) \le \frac{1}{8} (1+\epsilon)
\]
and hence
\[
\delta < \frac{1}{8} \frac{1+\epsilon_3}{1-\epsilon_3} < 1/7
.
\]
Finally,
\[
\epsilon \le \frac{\delta}{32}\frac{(1+\epsilon)^3}{1-\epsilon} <
	\changed{1.016,112\cdots \times \frac{\delta}{32}
	< \frac{\delta}{31}}
\]
\end{proof}

\begin{lemma}\label{lemma-C1C2}
	Let $\epsilon < \changed{\delta/31}$ with $\delta \le 1/7$. 
The values of $C_1$ and $C_2$ computed by any weak 
$\epsilon$-computation for the machine $U$ satisfy the inequalities below:
\[
\frac{1+\epsilon}{1-\delta/2} < C_1 < \frac{1-\epsilon}{1+\epsilon}
\ \frac{1}{1-\delta}
\]
and
\[
\frac{1+\epsilon}{1-\epsilon}\ \frac{1}{1+\delta} < C_2 < 
\frac{1-\epsilon}{1+\delta/2}
.\]
\end{lemma}

The proof of this Lemma uses $1+\epsilon$ inequalities. It is postponed to Appendix~\ref{sec-technical}.

\begin{proof}[Proof of Proposition~\ref{prop-technical}]
Assume the hypotheses of item (1). Lemmas~\ref{lemma-delta}
and~\ref{lemma-epsilon} guarantee that the tests in lines
2 and 3 pass, and in particular that the hypotheses of Lemma~
\ref{lemma-C1C2} are satisfied.
\bigskip
\par
\noindent
{\bf Item (a):} We assumed that $w$ was an accepting weak 
$\delta/2$-computation for $C$ with
input $(x, \delta)$, $\delta<1/8$, 
and that $w_i \in {\mathbf F}_{\epsilon} \subseteq {\mathbf F}_{\delta}$.
Property \eqref{a-ref}
guarantees that a stong computation for $U$ reads 
$w_i = \fl_{\epsilon}(w_i)$ exactly. 
\medskip
\par
For each coordinate $w_i >0$ of $w$ corresponding to reading input or
constant $c$, we have
\[
\frac{1}{1+\delta/2} w_i \le c \le \frac{1}{1-\delta/2} w_i
.
\]
Using the fact that $c(1-\epsilon) \le \fl_{\epsilon}(c) \le c(1+\epsilon)$,
we deduce that
\[
\frac{1-\epsilon}{1+\delta/2} w_i \le \fl_{\epsilon}(c) 
\le \frac{1+\epsilon}{1-\delta/2} w_i
.
\]
From Lemma~\ref{lemma-C1C2},
\[
C_2 w_i \le \fl_{\epsilon}(c) \le C_1 w_i
.
\]
Recall that $C_1$, $C_2$ and the $w_i$ are in ${\mathbf F}_{\epsilon}$.
By monotonicity,
\[
\fl_{\epsilon}(C_2 w_i) \le \fl_{\epsilon}(c) \le \fl_{\epsilon}(C_1 w_i) 
\]
so the test in line 8 passes. The same argument is valid for negative
$w_i$, reversing inequalities where necessary.
\medskip
\par
For each positive 
coordinate of $w_i$ of $w$ corresponding to an arithmetic operation,
we also have:
\[
\frac{1}{1+\delta/2} w_i \le w_j \circ_i w_k \le \frac{1}{1-\delta/2} w_i
\]
and hence from Lemma~\ref{lemma-C1C2},
\[
C_2 w_i \le w_j \circ_i w_k \le C_1 w_i
\]
Recall that a number $x$ is representable if and only if $x=\fl_{\epsilon}(x)$.
All the numbers $w_r$ in the inequation above are representable.
By monotonicity \eqref{a-mon}, we obtain that
\[
\fl_{\epsilon}( C_2 w_i ) \le \fl_{\epsilon}( w_j \circ_i w_k )
\le \fl_{\epsilon}( C_1 w_i).
\]
Therefore all the tests in line 11 pass. Reversing inequalities, all the
tests for negative values of $w_i$ also pass. 
\par
There is no rounding-off for selectors, so the tests in
line 14 pass.
Therefore the input is accepted. Since there are at most $2 \tau$ inequalities
to check, the machine terminates in time $\tau^{O(1)}$.

\bigskip
\noindent
\par
{\bf Item(b):} Assume that the input is accepted by a weak 
$\epsilon$-computation. Part 2 of Lemmas
~\ref{lemma-delta} and \ref{lemma-epsilon} guarantess that the hypotheses
of Lemma~\ref{lemma-C1C2} hold. In particular, $\delta < 1/7$.
\medskip
\par
It is convenient now to distinguish between the real values of
inputs $x,w,\delta,\epsilon$ and the values that are actually 
stored in memory during the weak $\epsilon$-computation.
Those will be denoted by $(\hat x, \hat w, \hat \delta, \hat \epsilon)$.
If $y$ are the constants of $C$, let $\hat y$ be the values stored in
memory during the weak $\epsilon$-computation. The symbols $C_1$ and $C_2$
stand for the values computed during the weak $\epsilon$-computation.
\par
We claim that $\hat w$ is a weak $\delta < \changedA{1/7}$ computation for the circuit
$C$ with input $x, \delta$. When $w_i>0$ corresponds to an input variable
or a constant (say $c$), line 8 guarantees that
\[
C_2 \hat w_i (1-\epsilon) \le \hat c \le C_1 \hat w_i (1+\epsilon)
\] 
so
\[
C_2 \hat w_i \frac{1-\epsilon}{1+\epsilon} 
\le c \le C_1 \hat w_i \frac{1+\epsilon}{1-\epsilon} .
\] 
and by Lemma~\ref{lemma-C1C2},
\[
\frac{1}{1+\delta} \hat w_i \le c \le 
\frac{1}{1-\delta} \hat w_i
.
\]
The same argument holds for $w_i < 0$ (line 9) and for variables corresponding
to a computation (lines 11--12). 
Finally, since $\delta < 1/7$, $\hat w$ is a weak accepting $\delta < \changedA{1/7}$-computation.
\end{proof}

\section{Transfer principle}

\begin{theorem}\label{real-transfer} If
$\P \ne \NP$,
then $\fpP \ne \fpNP$.
\end{theorem}

\begin{proof}[Proof of Theorem~\ref{real-transfer}]\label{ref23}
Assume that $\fpP = \fpNP$. Let $X$ be a problem in $\NP$. This means that
there are a machine $M = M(x,y)$ over ${\mathbb F}_2$ and a polynomial $p(t)$ 
so that
\begin{itemize}
\item [(a)] If $x \in X$, then there is $y \in \{0,1\}^{\infty}$ so that $M$ accepts the input $(x,y)$ within time $p(\Length{x})$.
\item [(b)] If there is $y \in \{0,1\}^{\infty}$ so that $M$ accepts the input $(x,y)$, then $x \in X$.
\end{itemize}
Define $X' \subseteq \mathbb R^{\infty}$ by $x' \in X'$ if and only if
$x \in X$, where $x_i = 1$ for $x_i' > 0$ and $x_i = 0$ for $x_i' \le 0$.
We claim that the problem $(X',1)$ is in $\fpNP$. Indeed, it is recognized
by the non-deterministic machine $M'$ that simulates the machine $M$.

But we assumed $\fpP = \fpNP$. Therefore $(X',1)$ is recognized by a deterministic polynomial time machine $P'$. Let $p_{\Prec}$ and $p_{\Arith}$ be the two polynomials associated to $P'$. 

The machine $P'$ operates with discrete and
continuous variables. But its internal constants are rational numbers. Let
	$\changedA{M'}$ be the machine over ${\mathbb F}_2$ that given an input 
$x \in \{0,1\}^{\infty}$, simulates
	the machine $\changedA{P'}$ with input $(x, 2^{-p_{\Prec}(\Length{x}}))$.
	$x \in \{0,1\}^{\infty}$. Then $\changed{M'}$ will accept $x \in X$ within
	time polynomial \changed{in} 
\[
p_{\Prec}(\Length{x}) p_{\Arith}(\Length{x}),
\]
and reject $x \not \in X$. Thus, $\P=\NP$.
\end{proof}

\begin{example}\label{real-encoding}\label{ref36a}
There is another possible encoding of discrete problems inside a continuous
problem.
	Let $M$ be an arbitrary \changed{BSS machine over ${\mathbb F}_2$}. Encode any arbitrary input $(x_1, \dots, x_N)$, $x_i \in \{0,1\}$ by 
	\changed{$x = \sum_{i=1}^N x_i 2^{-i}$} as explained in Definition~\ref{BSS}:
\[
x = (1,1,\dots,1,0.x_1, \dots, x_N).
\]
\changed{We choose\label{ref35} to define the condition number for the continuous
problem as $2^N$, so that the input size is mantained}.
	Let {\sc BitExpansion} be the routine from Example~\ref{ex-weak-strong}. Now consider the machine below:
\begin{trivlist}
\item {\tt Input}$(x,\epsilon)$.
\item $L \leftarrow$ {\sc BitExpansion}$(x)$.
\item {\bf If} $L<0$ {\tt Return} $-1$.
\item $(d,s,e,f_0,\dots, f_d) \leftarrow L$.
\item {\bf If} $e \ne 0$ or $s<0$ or $\epsilon > 2^{-d-1}$ {\tt Return} $-1$.
\item {\bf If} $f$ is not in the format above {\bf then} {\tt Return} $-1$.
\item Simulate the \changed{machine $M$} with input $x_1, \dots, x_N$.
Return $+1$ if the input is accepted, $-1$ otherwise.
\end{trivlist} 
	The pseudo-code above simulates the \changed{BSS machine $M$ over ${\mathbb F}_2$} in polynomial time
with respect to its own running time. It may fail if $\epsilon$ is too large,
but this is within the definition of $\fpP$. After obvious tweaks it turns
out that any set decidable in polynomial time within the Turing model and encoded as in Example~\ref{real-encoding} can
be recognized or decided in polynomial time in $\fpP$ (set the condition 
equal to $1$). 
\end{example}

\begin{remark} Does a reciprocal of Theorem~\ref{real-transfer} above hold?
Assume now that $\P=\NP$. We would like to claim that 
$(\CircPseudoFeas, \rho^{-1})$ can be decided in polynomial time.
\changedA{The idea would be to guess $w$ from \eqref{def-rho}.
The hypothesis $\P=\NP$ would allow us to conclude that $\fpP = \fpNP$.
Unfortunately, we do not know how to conveniently 
bound the bit size of $w$ in terms
of $\rho^{-1}$. This is because we do not know how to bound the bit size
of the {\em exponent} of each $w_i \in {\mathbf F}_{\epsilon}$ polynomially
	in $\Size{C}$.}
\end{remark}

\changedA{
\begin{remark}
	It is easy to embed every instance of $\CircPseudoFeas$ into an instance of $\SAFeas$. Indeed, let the circuit $C$ have input $x_1, \dots, x_n, \delta$ and constants $y_1, \dots, y_m$. For simplicity let us make $x_{n+1}=\delta$.
The defining equations for a $\delta$-pseudo orbit are:
\begin{align*}
	0 \le \delta &< 1/4 \\
	|w_i - x_i| &\le \delta |x_i| && i=1, \dots, n+1\\
	|w_i - y_i| &\le \delta |y_i| && i=n+2, \dots, m\\
	|w_i - w_{j(i)} \circ_i w_{k(i)}| &\le \delta |w_{j(i)} \circ_i w_{k(i)}| 
	&& \text{if $i$ is a computation node} \\
	w_i &= S(w_{j(i)}, w_{k(i)}, w_{l(i)}) && \text{if $i$ is a selector node.}
\end{align*}
Inequalities with absolute value can be squared to obtain algebraic inequalities. The inequalities with a selector can be replaced by a system of algebraic equalities
\begin{eqnarray*}
	s_i &=& w_{l(i)} u_i^2 \\
	(s_i - 1)(s_i+1) w_{l(i)} &=& 0 \\
	2t_i &=& s_i (s_i+1) \\
	w_i &=& t_i w_{j(i)} + (1-t_i) w_{k(i)}.
\end{eqnarray*}
	Above, $u_i, t_i, s_i$ are new variables. Notice that $s_i \in \{-1,0,1\}$ by construction and $t_i \in \{0,1\}$.
This reduction makes $(\SAFeas, 1)$ $\fpNP$-hard, but not necessarily
	$\fpNP$-complete. We could also define $\mu(C) = 2^{\min (\text{Bitsize}(w))}$,
	where the minimum is taken over all the values $w \in {\mathbf F}_{\epsilon}^{\Length{C}}$ that satisfy the equations above. With this definition we
would obtain a problem $(X, \mu)$ in $\fpNP$ for a subset $X \subset
\SAFeas$, but we do not know if the reduction is polynomial time (See the previous remark).
\end{remark}
}
\section{Conclusions and further comments}

In this paper we developed a complexity theory for numerical computations
satisfying the wish list of the introduction. However, we needed to give up on other
reasonable ‘wishes’ for a good theory of numerical computations.
\changed{
New candidates for `polynomial time' and
`nondeterministic polynomial time' were obtained here.
	
While it makes no sense to compare complexity classes from different models
of computation, one should certainly try to establish connections between
them. In particular, the `polynomial time' classes are supposed to model
problems considered easy by practitionners.
Different people may rightfully disagree about the most appropriate
model. 
We contend that the one-sided class $\realPone$ is a sensitive
choice for numerical mathematics as long as numerical 
stability is not an issue. This applies in particular to people plotting
Julia sets or solving some geometrical problems.
People doing algebraic computations may prefer the BSS model over
a ring $R$, so they will prefer some $\P_{\mathrm R}$, maybe $\realP$.
Of course classes $\realP$ and $\realPone$ are not comparable, but Main Theorem~A established a connection between the two theories.

The class $\fpP$ is an attempt to introduce numerical stability into the
picture. The condition on weak $\epsilon$-computations guarantees that a
machine recognizing some problem $(X, \mu)$ will never return a wrong answer
in the $1+\epsilon$ model of computation. This is one of the preferred 
models for numerical analysts yet it is too weak to establish many
useful and widely used properties of computer arithmetic. In particular
we are unable to check if $a>b+c$ for approximations $a$, $b$ and $c$, 
or to build \changedA{a} universal machine.

Then there are strong $\epsilon$-computations that model the IEEE floating
point standard. Here, classical complexity theory (BSS over ${\mathbb F}_2$
or equivalently Turing complexity) hits back. If one defines 
a class of
polynomial time computations in the IEEE or in the strong computation model
(\ocite{Priest}'s model), then one produces essentially the class $\P$ from
classical complexity for input in a very special encoding. Numerical
stability is no more an issue, it is bypassed by the fact that IEEE
computations are discrete computations. The definition of $\fpP$ in this
paper was a reasonable compromise, fulfilling wishes (d), (e) and (f)
from our list.

It is important to mention that classes such as $\fpP$ and $\fpNP$ 
may be relevant to people solving \changedA{seemingly intractable 
computational problems.
Most numerical
analysts deal with `tractable' problems, where concepts such 
as `polynomial time'
are too crude: reducing the exponent of a polynomial time
bound, or even reducing the constant are more practically important matters.
They may also disregard the worst case. But once the problem is intractable
enough (e.g. non-linear optimization) the polynomial hierarchy becomes
relevant.}

The known connection between classes $\fpP$, $\fpNP$ and their classical 
counterpart is Main Theorem B. There seems to be no close connection 
between $\realPone$ and $\fpP$.

Another theoretical issue is whether problems in $\fpNP$ can be solved in
exponential time. 
}
While it was easy to check that $\realNPone \subset \realEXPone$,
we do not know if \changed{$\fpNP \subset \fpEXP$}\label{ref36b}.


Also, it would be desirable to have a cleaner, more natural $\fpNP$-complete 
problem such as the traveling salesman problem in classical complexity or 
the Nullstellensatz in BSS complexity. 

The proof that $\CircPseudoFeas$ is $\fpNP$-complete seems to largely rely on the fact that the input of a floating point number in floating point with same or more precision is exact and on relative error estimates. We expect that a simpler theory relying on relative error (weak computations) alone not on the other more subtle properties of computations may be developed with $\NP$-complete problems. One approach would be \changed{to assume}\label{ref36c} that the input is always exact. The condition number of \CircPseudoFeas \ might need to be adjusted.

Those and other issues are left for future investigations.

\appendix
\section{Proof of technical results}
\label{sec-technical}
\renewcommand{\thetheorem}{\ref{lemma-sum}}
\begin{lemma}
\changed{Assume radix $\beta = 2$.}
Let $a, b \in {\mathbf F}_{\epsilon}$, $|a|\ge |b|$.
Let $c=\fl_{\epsilon}(a+b)$,
	$d=\fl_{\epsilon}(c-a)$, $e=\fl_{\epsilon}\changed{(b-d)}$. Then
$a+b = c+e$ exactly.
\end{lemma}

The proof needs some fine properties of the floating point number system.
Those are present in most computer arithmetics, but in this paper they
follow from the construction of floating point numbers.
\begin{eqnarray}
\label{a-Sterbenz}
\forall a,b \in {\mathbf F}_{\epsilon}, a \ge 0.
&&
\frac{1}{2}a \le b\le 2a \Rightarrow a-b \in {\mathbf F}_{\epsilon}
\\
\label{a-A1}
\forall a,b \in {\mathbf F}_{\epsilon},
&& a+b-\fl_{\epsilon}(a+b)
\in {\mathbf F}_{\epsilon}.\\
\label{a-A2}
\forall a,b \in {\mathbf F}_{\epsilon},
&&
|b|\le |a|
\Rightarrow |\fl_{\epsilon}(a+b)| \le 2|a|.
\end{eqnarray}
Those are respectively
Sterbenz's Lemma p.12 and Properties A1, A2 p.17
in \cite{Priest}. \changed{Notice that equation \eqref{a-A2} does not hold
with radix $\beta=10$, take for instance $a,b=0.999$ and three digits of
mantissa, so $\fl(a+b)=2.00 > 1.998$.}
Also by construction if $a \in {\mathbf F}_{\epsilon}$
and the radix $\beta$ is $2$ then $a/2 \in {\mathbf F}_{\epsilon}$. If 
$a \in {\mathbf F}_{\epsilon}$,
$\beta \ne 2$ 
and $a/2 \not \in {\mathbf F}_{\epsilon}$ then $a/2$ is equidistant to
two representable numbers.

\begin{proof}[Proof of Lemma~\ref{lemma-sum}]
Assume withouth loss of generality that $a \ge 0$.
There are two cases. 
\\
	{\bf Case 1:} Suppose that $a/2 \le -b$.
Property
\eqref{a-Sterbenz} guarantees that $c$ was computed exactly,
so $d=b$ and $e=0$. 
\\
{\bf Case 2:} \changed{Suppose now that $a/2 > -b$. We claim
that $a/2 \le c$. Indeed, $-a/2 < b < 0$ and hence
$a+b > a-a/2=a/2$.
The monotonicity property \eqref{a-mon} gives us the inequality
$c = \fl_{\epsilon}(a+b) \ge \fl_{\epsilon}(a/2)$.
Because the radix is $\beta=2$ and $a \in {\mathbf F}_{\epsilon}$,
$a/2$ is representable whence $a/2 \le c$.}
	
From \eqref{a-A2} and the hypothesis $|b| \le |a|$ we conclude
that $c=|c|=\fl_{\epsilon}(a+b) \le 2 |a|=2a$. 
	Therefore \changed{$0 \le c/2 \le a \le 2c$ and
Property \eqref{a-Sterbenz} implies that  
	$c-a \in {\mathbf F}_{\epsilon}$. Therefore,} 
$d=c-a$ exactly. Then \eqref{a-A1} 
	\changed{implies that
		$b-d=(a+b)-c= (a+b)-\fl_{\epsilon}(a+b) \in {\mathbf F}_{\epsilon}$
	so $e=(a+b)-c$} exactly.
\end{proof}

\changed{
\renewcommand{\thetheorem}{\ref{abc}}
\begin{lemma} Assume radix $\beta=2$.
Let $a, b, c \in {\mathbf F}_{\epsilon}$
be strictly positive, with $b \ge c$. Let $d=\fl_{\epsilon}(a-b)$ and
$e=\fl_{\epsilon}(d-c)$. Then,
\[
\sgn(e)
=
\sgn(a-b-c)
\]
\end{lemma}
\begin{proof}
The proof is divided into three cases.
\\
{\bf Case 1:} If $b > 2a$ then $d=\fl_{\epsilon}(a-b)<0$ and hence $e<0$. 
\\
{\bf Case 2:} If $a/2 \le b \le 2a$, then Sterbenz property \eqref{a-Sterbenz}
implies that $a-b \in {\mathbf F}_{\epsilon}$, hence $d=a-b$ exactly.
Then $e = \fl_{\epsilon}( a-b-c)$, so $\sgn(e)=\sgn(a-b-c)$.
\\
{\bf Case 3:}
If $b < a/2$, then from the hypotheses $a-b-c > 0$.
Also, $a-b > a/2$ and $d= \fl_{\epsilon}(a-b) 
\ge \fl_{\epsilon}(a/2)$ because of property \eqref{a-mon}.
The number $a$ is representable in ${\mathbf F}_{\epsilon}$
and the radix is $\beta=2$,
	therefore $\fl_{\epsilon}(a/2)=a/2$ and $d \ge a/2$. We have
\[
c \le b < a/2 \le d.
\]
It follows that $d-c > 0$ hence $e= \fl_{\epsilon}(d-c) > 0$.\\
\end{proof}
}

\renewcommand{\thetheorem}{\ref{lemma-C1C2}}
\begin{lemma}
	Let $\epsilon < \changed{\delta/31}$ with $\delta \le 1/7$. 
The values of $C_1$ and $C_2$ computed by any weak 
$\epsilon$-computation for the machine $U$ satisfy the inequalities below:
\[
\frac{1+\epsilon}{1-\delta/2} < C_1 < \frac{1-\epsilon}{1+\epsilon}
\ \frac{1}{1-\delta}
\]
and
\[
\frac{1+\epsilon}{1-\epsilon}\ \frac{1}{1+\delta} < C_2 < 
\frac{1-\epsilon}{1+\delta/2}
.\]
\end{lemma}

\begin{proof}[Proof of Lemma~\ref{lemma-C1C2}]
By the $1+\epsilon$ inequality, 
\[
\left(1+\frac{3}{4}\delta(1-\epsilon)^2\right)(1-\epsilon)^2
\le
C_1 \le \left(1+\frac{3}{4}\delta(1+\epsilon)^2\right)(1+\epsilon)^2
\]
and
\[
\left(1-\epsilon-\frac{3}{4}\delta(1+\epsilon)^3\right)(1-\epsilon)
\le
C_2
\le
\left(1+\epsilon-\frac{3}{4}\delta(1-\epsilon)^3\right)(1+\epsilon)
.
\]
After replacing $C_1$ and $C_2$ by the bounds above and
clearing denominators, it
suffices to check the following inequalities:
\begin{eqnarray*}
(1-\delta)
\left(1+\frac{3}{4}\delta(1+\epsilon)^2\right)(1+\epsilon)^3
-
(1-\epsilon)
&<& 0 \\
(1-\delta/2)
\left(1+\frac{3}{4}\delta(1-\epsilon)^2\right)(1-\epsilon)^2
-
(1+\epsilon)
&>& 0 \\
(1+\delta/2)\left(1+\epsilon-\frac{3}{4}\delta(1-\epsilon)^3\right)(1+\epsilon)
-
(1-\epsilon)
&<& 0 \\
(1+\delta)\left(1-\epsilon-\frac{3}{4}\delta(1+\epsilon)^3\right)(1-\epsilon)^2
-
(1+\epsilon)
&>&0
\end{eqnarray*}
The left-hand-sides of those inequations are respectively increasing, decreasing, increasing and decreasing with respect to $\epsilon$.
Therefore it is enough to check the inequalities with $\epsilon=\delta/48$.
The constant terms vanish. After dividing by $\delta$, one obtains the
following inequalities to check:
\changed{
		\begin{eqnarray*}
-{{3\,\delta^6}\over{114516604}}-{{231\,\delta^5}\over{57258302}}- {{915\,\delta^4}\over{3694084}}-{{226\,\delta^3}\over{29791}}-{{ 13853\,\delta^2}\over{119164}}-{{1389\,\delta}\over{1922}}-{{15 }\over{124}}
&<&0\\
-{{3\,\delta^5}\over{7388168}}+{{189\,\delta^4}\over{3694084}}-{{ 291\,\delta^3}\over{119164}}+{{101\,\delta^2}\over{1922}}-{{3371\, \delta}\over{7688}}+{{19}\over{124}}
&>&0\\
{{3\,\delta^5}\over{7388168}}-{{45\,\delta^4}\over{1847042}}-{{3\, \delta^3}\over{59582}}+{{95\,\delta^2}\over{3844}}-{{2255\,\delta }\over{7688}}-{{19}\over{124}}
&<&0\\
-{{3\,\delta^6}\over{114516604}}+{{231\,\delta^5}\over{57258302}}- {{915\,\delta^4}\over{3694084}}+{{224\,\delta^3}\over{29791}}-{{ 13117\,\delta^2}\over{119164}}+{{1029\,\delta}\over{1922}}+{{201 }\over{124}}
&>&0
		\end{eqnarray*}}
Those can be established numerically
for $\delta \in [0,1/7]$. 
\end{proof}

\section{Alternative construction of the circuit of 
\changedA{Theorem~\ref{CircFeasNPc} and} Proposition~\ref{th-main}}
\label{alt-reduction}
\begin{figure}
\centerline{\resizebox{\textwidth}{!}{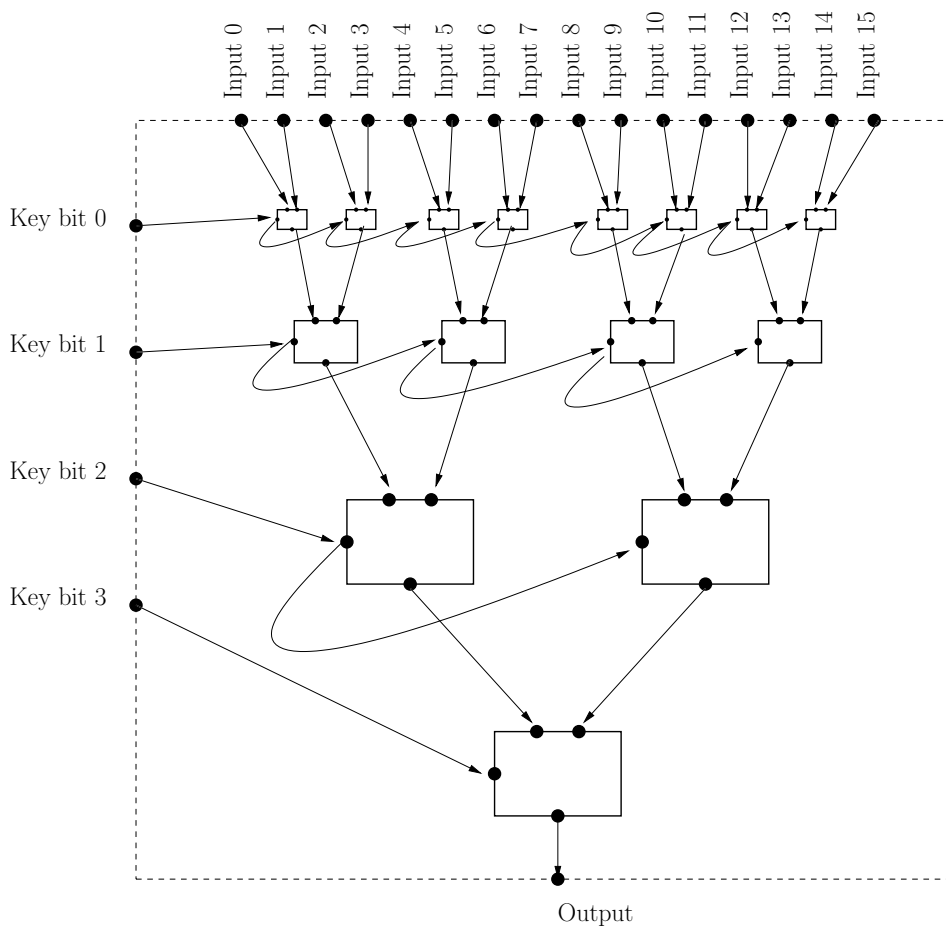}}
\caption{Elementary selectors can be arranged in order to build
a compound selector switching between at most $2^d$ inputs. This
selector is controlled by a 
$d$ bits key in binary representation.\label{fig-switch}}
\end{figure}
\begin{figure}
\centerline{\resizebox{\textwidth}{!}{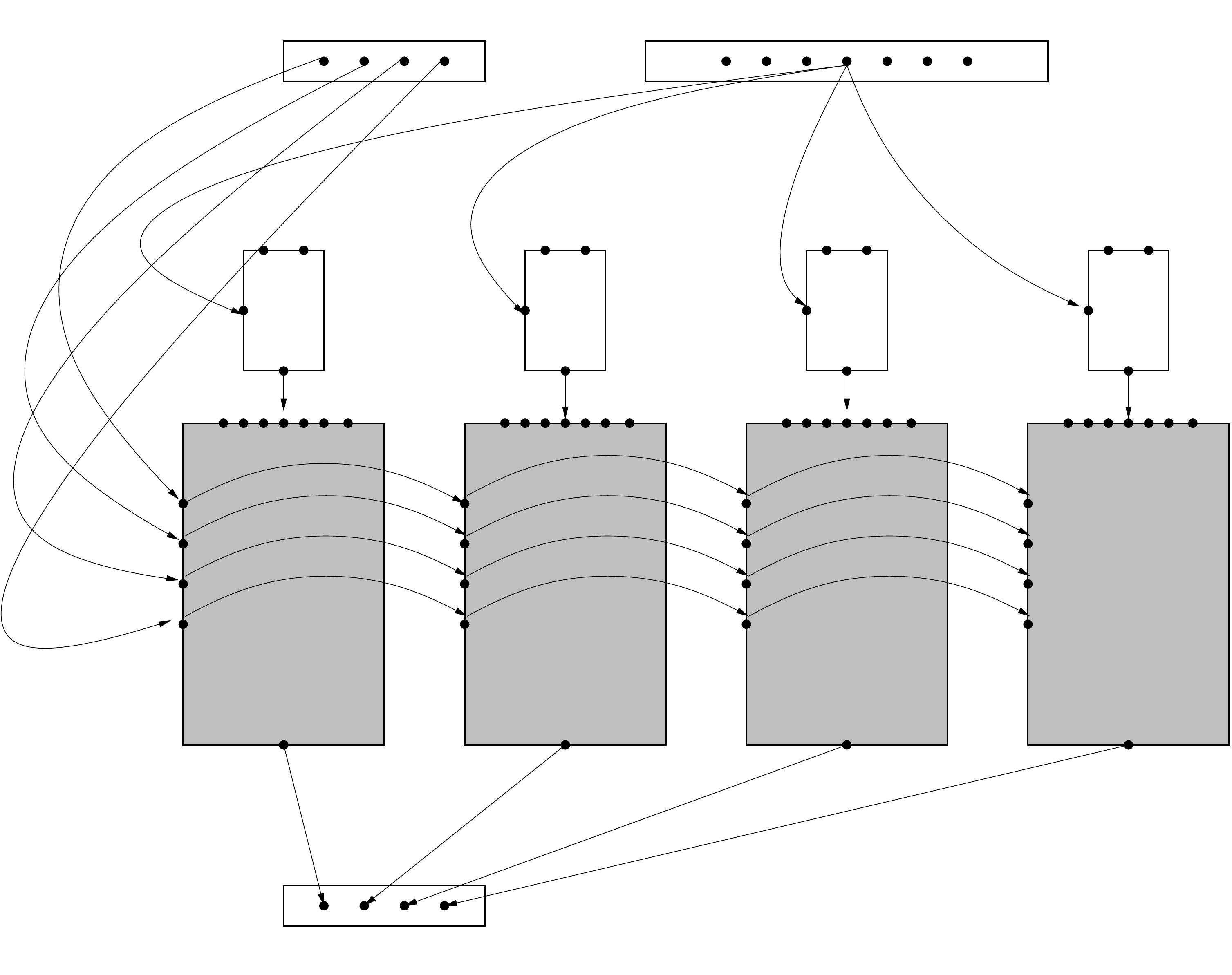}}
\caption{A compound selector from Fig.~\ref{fig-switch} can be used
 to compute each bit of $\nu(t+1)$. One
should feed the $i$-th input of the $k$-th compound selector with 
the $k$-th bit of $\beta(i,s_0)$. If $i$ is a branching node,
this requires an extra selector controlled by $s_0(t)$. 
\label{fig-next-nu}
}
\end{figure}

As before, $(\nu(t), s(t))_{t \in \mathbb N_0} \in \{1,2, \dots, N \} 
\times \mathbb R_{\infty}$
denote the {\em node} and {\em state} at time $t$.
The node $\nu(t)$ will be represented by $d=\lceil \log_2 N \rceil$
`discrete' variables in $\{0,1\}$, that is $d$-bits. Each bit of the 
next node $\nu(t+1)$
can be obtained through the digital circuit in Figure~\ref{fig-switch},
where the input $i$ is \changed{fed}\label{ref39} with the corresponding bit of
the selector output,
\[
	\beta(i, s_0(t))=S(s_0(t),\beta^+(i)\changed{)},
\beta^-(i)
\]
(Please see figure~\ref{fig-next-nu}). 
The same trick is used to compute each coordinate of the next
state $s_k(t+1)$:
we add a circuit for each $g(\nu, s(t-1))$, $\nu=1, \dots N$,
and then a binary switching tree to compute $g(\nu,s)$.
For $k \ne 0$ there are only three possibilities $s_k(t+1) = s_{k-1}(t)$,
or $s_k(t)$, or $s_{k+1}(t)$. There can be binary arithmetic operations
or constants when $k=0$.
This circuit will not introduce new numerical 
computations except those that where already in $M$.

In other words, a weak (resp. strong, resp. exact) 
time-$T$ $\delta$-computation for $M$ maps
to a weak (resp. strong, resp. exact) $\delta$-computation for 
$C_{M,T,x}$. The reciprocal is true:
any weak (resp. strong, resp. exact) $\delta$-computation for 
$C_{M,T,x}$ can be written as a subsequence of a
weak (resp. strong, resp. exact) 
time-$T$ $\delta$-computation for $M$, discrete variables removed.

\end{document}